\documentclass[11pt]{article}

\usepackage{algorithm}
\usepackage{algpseudocodex,float}
\usepackage{lipsum}
\usepackage{url}
\usepackage{graphicx}
\usepackage{amsmath}
\usepackage{amsthm}
\usepackage{indentfirst}
\usepackage{cases}
\usepackage{xcolor}
\usepackage{comment}
\newcommand{\trfrcost}{\lambda}
\newcommand{\hyperpar}{\alpha}
\usepackage{times}
\usepackage{anysize}
\marginsize{1in}{1in}{1in}{1in}

\usepackage{mathtools}

\usepackage{enumitem}
\setlist[enumerate]{listparindent=\parindent}

\usepackage{balance}

\AtBeginDocument{\pagenumbering{arabic}}
\allowdisplaybreaks[4]
\makeatletter

\newtheorem{Lemma}{Lemma}
\newtheorem{observation}{Observation}
\newtheorem{Proposition}{Proposition}

\begin{document}
\title{Cost-Driven Data Replication with Predictions}

\author{
        Tianyu Zuo, Xueyan Tang, and Bu Sung Lee\\
        Nanyang Technological University\\
        Singapore\\
        zuot0001@e.ntu.edu.sg, \{asxytang,~ebslee\}@ntu.edu.sg
}
\date{}
\maketitle

\begin{abstract}
This paper studies an online replication problem for distributed data access. The goal is to dynamically
create and delete data copies in a multi-server system as time passes to minimize the
total storage and network cost of serving access requests. We study the problem in the emergent learning-augmented setting, assuming simple binary predictions about inter-request times at individual servers. We develop an online algorithm and prove that it is $\big(\frac{5+\hyperpar}{3}\big)$-consistent (competitiveness under
perfect predictions) and $\big(1+\frac{1}{\hyperpar}\big)$-robust (competitiveness under terrible predictions), where $\hyperpar\in (0, 1]$ is a hyper-parameter representing the level of distrust in the predictions. We also study the impact of mispredictions on the competitive ratio of the proposed algorithm and adapt it to achieve a bounded robustness while retaining its consistency. We further establish a lower bound of $\frac{3}{2}$ on the consistency of any deterministic learning-augmented algorithm. Experimental evaluations are carried out to evaluate our algorithms using real data access traces. 
\end{abstract}

\section{Introduction}
Online decision-making problems are an important field of algorithmic research. These problems seek to make a chain of judicious decisions as time moves on, without the knowledge of what will happen in the future. The data replication (or caching) problem for distributed data access at various locations is one such example with many application scenarios, such as cloud storage services, content distribution networks, edge and fog computing. The goal of this problem is to store data copies in a distributed system, so that the requests for data access can be satisfied with minimal service costs. The service costs normally include the storage and network costs, because storing data copies in servers and transferring data among servers both consume resources and incur expenses for service providers \cite{2017Data, 2018Fog}. In this paper, we study dynamic replication in a distributed system to minimize the total storage and network cost for serving a sequence of data access requests, where the storage cost is modelled proportional to the time duration of storage and the network cost is modelled proportional to the amount of network traffic.
Unaware of the access times and locations of future requests, an online algorithm of managing data copies is needed to minimize the cumulative service cost over the time span.

One paradigm of tackling uncertainty in future inputs is competitive analysis, which refers to the idea of guaranteeing a bound on the performance of an online algorithm over all instances of the problem. To this end, a metric called {\em competitive ratio} is used to measure the ratio between the costs produced by an online algorithm (without knowledge of future) and the optimal offline solution (with full knowledge of future) in the worst case \cite{borodin1998, vazirani2013approximation}. Some researches have been carried out to study the data replication problem with competitive analysis. Mansouri and Erradi \cite{mansouri2018cost}  developed online algorithms for a two-tier storage system which consists of a hot tier with high storage cost but low access cost, and a cold tier with low storage cost but high access cost. Mansouri {\em et al.} \cite{mansouri2019cost}  developed online algorithms for data replication and migration in cloud storage services with competitive ratios dependent on the max/min cost ratio among data centers.
Veeravalli \cite{veeravalli2003network} presented a general model for migrating and caching shared data in a network of servers.
Assuming the storage costs of all servers are identical, Bar-Yehuda {\em et al.} \cite{bar2012growing} developed an $O(\log \delta)$-competitive online algorithm where $\delta$ is the normalized diameter of the underlying network.
Under the same assumption of identical storage costs of all servers, Wang {\em et al.} \cite{wang2018cost} proposed an optimal offline solution by dynamic programming and a $3$-competitive online algorithm in the case of uniform transfer costs among servers. Recently, they also extended the solutions to servers with distinct storage costs and presented an online algorithm claimed to be $2$-competitive \cite{2021Cost} (the claim however is not true, as shall be explained in Section \ref{remarks}).

Traditional competitive analysis, however, is often too pessimistic in that it assumes no knowledge about future inputs at all. With the rapidly growing power of machine learning technologies, the emergence of predictive models in recent years provides another possible way of beating uncertainty in online decision-making problems. By forecasting the future based on historical data, predictive models offer additional information for decision-making. We can use the predictions about future inputs to design more effective online algorithms both in theory and practice. Nevertheless, predictive models are usually imperfect and prone to errors. Unconditional reliance on the forecasts of predictive models may on the contrary deteriorate online algorithms \cite{mitzenmacher2022algorithms}. Incorrect or inaccurate predictions may mislead online algorithms to make wrong decisions and result in unbounded competitive ratios. Hence, machine-learned predictions should be wisely incorporated into the algorithm design, so that both prediction use and worst-case guarantee can be achieved.

Some recent studies have explored augmenting online algorithms with predictions to address the classical ski-rental problem \cite{1988Com}.
Kumar \textit{et al.} \cite{2018kumar} proposed a deterministic algorithm and a randomized algorithm with the use of predictions. By introducing a hyper-parameter to manipulate the reliance on the predictions, their algorithms achieve a tight trade-off between the best case of perfect predictions and the worst case of terrible predictions \cite{2020optimal}.
Gollapudi and Panigrahi \cite{2019online} further considered a more general case where multiple predictions are provided by a collection of machine learning predictors. Based on the distribution of the predictions, their algorithm carefully determines when to rent or buy skis. Kodialam \cite{2019k} studied another scenario of using  machine-learned predictions, where the predictions are in the probabilistic form rather than the deterministic form. He developed a randomized algorithm with the aid of such soft machine-learned predictions.

Our work is inspired by the above studies.
As shall be elaborated later, due to the trade-off between storage and network costs, making a replication decision at a single server in our context of distributed data access resembles the ski-rental problem. But our problem is much more complicated than ski-rental in that it involves multiple servers and multiple requests with inter-dependencies.
To the best of our knowledge, there has been no study on making sensible use of predictions for cost-driven data replication in a distributed system.
We aim to develop a learning-augmented online algorithm to optimize the cumulative cost of serving a request sequence.

Our contributions are summarized as follows.
\begin{enumerate}
    \item We propose an online replication algorithm which uses the forecasts provided by predictors, in order to minimize the total cost of serving a request sequence.
    \item We conduct theoretical analysis to show that the proposed algorithm is $\big(\frac{5+\hyperpar}{3}\big)$-consistent (competitive ratio under perfect predictions), and $\big(1+\frac{1}{\hyperpar}\big)$-robust (competitive ratio under terrible predictions), where $\hyperpar\in (0, 1]$ is a hyper-parameter representing the level of distrust in the predictions. We also study the impact of mispredictions on the competitive ratio and adapt our algorithm to achieve a bounded robustness while retaining its consistency.
    \item We establish a lower bound of $\frac{3}{2}$ on the consistency of any deterministic learning-augmented algorithm. This implies that it is not possible to achieve a consistency approaching $1$ in our problem. 
    \item We empirically evaluate our algorithm using real data access traces, and show that it can make effective use of predictions to improve performance with increasing prediction accuracy.
\end{enumerate}

\section{Problem Definition}
We consider a system including $n$ geo-distributed servers (or sites) $s_{1},s_{2},$ $\dots,s_{n}$. A data object is hosted in the system and copies of this object can be created and stored in any servers.\footnote{We do not consider any capacity limit on creating data copies in servers, since storage is usually of large and sufficient
capacity nowadays. Hence, we focus on the management of one data object, as different objects can be handled separately.}  Storing a data copy in a server incurs a cost of 1 per time unit. The data object can also be transferred between servers when needed. The transfer cost of the object between any two servers is $\trfrcost$ $(\trfrcost > 0)$.

Requests to access the data arise at the servers over time (e.g., due to the computational jobs running at the servers or user requests). When a request arises at a server $s_j$, if $s_j$ holds a data copy, the request is served locally. Otherwise, the object has to be transferred from a server holding a copy to $s_j$ in order to serve the request.
$s_j$ may, as a result of the transfer, create a data copy and hold it for some period of time. We denote the sequence of requests arising at the servers as $\langle r_1, r_2, r_3, \dots\rangle$. For each request $r_i$, let $s[r_i]$ denote the server at which $r_i$ arises and $t_i$ denote the time when $r_i$ arises. For simplicity, we assume that all requests arise at distinct time instants, i.e., $t_1<t_2<t_3<\cdots$.
In addition, we assume only one data copy placed in server $s_1$ initially. To facilitate algorithm design and analysis, we add a dummy request $r_0$ arising at server $s_1$ at time $0$. Note that $r_0$ does not incur any additional cost for serving the request sequence.

We seek to develop a \textit{replication strategy} that determines the data copies to create and hold as well as the transfers to carry out in an \textit{online} manner to serve all requests. There must be at least one data copy stored at all times to preserve the data. Our objective is to minimize the total storage and transfer cost of serving a request sequence. We focus on the problem in the learning-augmented setting and assume there are predictions of inter-request times at individual servers (e.g., based on the request history or other features). Specifically, we assume that after a request arises at a server, a simple binary prediction is available about whether the next request at the same server will arise within or beyond a period of $\trfrcost$ time units from the current request. Since there are $n$ servers, we have $n$ predictions at any time, each for the next request at one server. We would like to make sensible use of the predictions to reduce the cost for serving requests.

For an online algorithm with predictions, the competitive ratio is often expressed as a function of prediction errors. As two special cases, when all predictions are perfect (zero error), the competitive ratio is known as  \textit{consistency}, whereas when the predictions can be arbitrarily bad (unbounded errors), the competitive ratio is known as  \textit{robustness} \cite{2018kumar}.

\section{Algorithm Design}
\label{design}
We start with investigating the problem structure. Observe that there is a trade-off between holding and not holding a data copy at each individual server.
If no copy is held by a server, we have to pay the transfer cost for serving each local request. If a copy is held by a server, we have to pay the storage cost which is proportional to the duration of storage. Such a trade-off appears to resemble that between rent and buy in the classical ski-rental problem \cite{1988Com}. Without predictions of inter-request times, to achieve decent competitiveness, an intuitive idea is to let a server $s_j$ hold a copy for a period of $\trfrcost$ time units after serving every local request. If the next request arises at $s_j$ before this period ends, the request will be served by the local copy and we pay the storage cost which is optimal. After that, $s_j$ will renew the copy and hold it for another period of $\trfrcost$ long. If no request arises at $s_j$ during the period of $\trfrcost$, $s_j$ will delete the copy and an incoming transfer will then be required to serve the next request at $s_j$. In this case, the total cost paid is $2\cdot\trfrcost$, while the optimal cost is a transfer cost $\trfrcost$ (by not holding a copy in $s_j$).
On the other hand, with perfect predictions of inter-request times, we can always achieve the optimal cost. That is, if the inter-request time between two consecutive requests at a server $s_j$ is no more than $\trfrcost$, we keep a data copy in $s_j$ between the requests and pay the storage cost which is optimal. Otherwise, we do not keep a copy in $s_j$ and pay the transfer cost for serving the latter request which is optimal.  
Nevertheless, this solution does not have bounded robustness, because (1) if the inter-request time is predicted to be longer than $\trfrcost$ but is actually shorter than $\trfrcost$, the transfer cost to pay is fixed while the optimal cost can be close to $0$; and (2) if the inter-request time is predicted to be shorter than $\trfrcost$ but is actually longer than $\trfrcost$, the storage cost can go towards infinity in the worst case. 
Furthermore, if we simply apply the above ideas (with or without predictions) to every server, it may not meet the requirement of maintaining at least one data copy at any time, because all copies will be deleted after a sufficiently long silent period without any request.

Our approach to algorithm design is to integrate and balance between following predictions and not using predictions. We introduce a hyper-parameter $\hyperpar\in (0, 1]$ to indicate the level of distrust in the predictions. $\hyperpar \rightarrow 0$ indicates nearly full trust and reliance on the predictions, while $\hyperpar \rightarrow 1$  indicates almost no trust and not using the predictions. If the next request at a server is predicted to arise beyond a period of $\trfrcost$ from the current request, we let the server hold a copy for a period of $\hyperpar\cdot\trfrcost$ instead of deleting the copy immediately. Then, even if the prediction is not correct, the next request still has some chance to be served by the local copy, avoiding the transfer. This will save the cost when the next request actually arises quite soon and thus improve robustness while keeping the loss in consistency under control. 
On the other hand, if the next request at a server is predicted to arise within a period of $\trfrcost$ from the current request, we let the server hold a copy for a period of $\trfrcost$ rather than up to
the next request. Then, even if the prediction is not correct, the storage cost will not increase infinitely, thereby enhancing robustness while still ensuring consistency. In addition, to meet the at-least-one-copy requirement, the copy whose period ends the latest according to the above strategies will continue to be kept beyond its period.

Algorithm \ref{alg3} shows the details of our proposed algorithm.
We use $E_j$ to denote the intended expiration time of the data copy in server $s_j$,
use $K_j$ to denote a binary tag indicating whether server $s_j$ keeps a copy beyond the intended expiration time,
and use $c$ to denote the number of servers holding data copies. Initially, there is only one data copy in server $s_1$, so $c=1$ (line 2).

\begin{algorithm}[t]
  \caption{Dynamic Replication with Predictions}
  \label{alg3}
  \begin{algorithmic}[1]
  \State /* a data copy is initially stored in server $s_{1}$ */
  \State
  $\textbf{initialize:}$ $c \gets 1$; $E_1 \gets \trfrcost$ or $\hyperpar\cdot\trfrcost$ ; $E_j \gets -\infty$ for all $2 \leq j \leq n$; $K_j \gets 0$ for all $1 \leq j \leq n$; \Comment{$s_1$ holds a data copy}
  \While {(a request $r_i$ arises at server $s_{j}$ at time $t_{i}$)}
    \If {$t_i \leq E_j$ or $K_j = 1$} \Comment{$s_j$ holds a data copy}
        \State serve $r_{i}$ by the local copy in $s_{j}$;
    \Else
        \State serve $r_{i}$ by a transfer from any other server with a copy;
        \State create a copy in $s_j$;
        \State $c \gets c+1$;
    \EndIf
    \If {the prediction forecasts the next request at $s_j$ will arise no later than time $t_i + \trfrcost$}
        \State $E_j \gets t_{i}+\trfrcost$;
    \Else
        \State $E_j \gets t_{i}+\hyperpar\cdot\trfrcost$;
    \EndIf
    \State $K_j \gets 0$;
  \EndWhile
  \While {($s_{j}$ transfers the data object to another server $s_{k}$ at time $t$)}
    \If {$K_j = 1$} \Comment{$s_{j}$ holds the only copy} 
        \State drop the copy in $s_{j}$;
        \State $K_j \gets 0$;
        \State $c \gets c-1$;
    \EndIf
  \EndWhile
  \While {(a copy expires in server $s_{j}$ at time $E_j$)}
    \If {$c=1$} \Comment{$s_{j}$ holds the only copy}
        \State $K_j \gets 1$;
    \Else
        \State drop the copy in $s_{j}$;
        \State $c \gets c-1$;
    \EndIf
  \EndWhile
 \end{algorithmic}
\end{algorithm}

When a request arises at a server $s_j$, if $s_j$ holds a data copy, the request is served locally (lines 4-5). Otherwise, the request is served by a transfer from another server with a copy (lines 6-9). After serving the request, $s_j$ keeps the data copy for an {\em intended duration} based on the prediction of the following inter-request time. If the prediction forecasts the next request at $s_j$ to arise within a period of $\trfrcost$, $s_j$ keeps the copy for an intended duration of $\trfrcost$. Otherwise, $s_j$ keeps the copy for an intended duration of $\hyperpar\cdot\trfrcost$ (lines 10-13). In either case, the tag $K_j$ of $s_j$ is cleared (line 14).\footnote{
With the dummy request $r_0$ arising at server $s_1$ at time $0$, $s_1$ sets an intended duration $\trfrcost$ or $\hyperpar\cdot\trfrcost$ for the initial copy according to the prediction of the first request arrival (line 2).}

In the intended duration, if a new request arises at $s_j$, the request is served by the local copy and $s_j$ renews the copy for another intended duration based on the new prediction of the subsequent inter-request time.
When the intended duration of the copy in $s_{j}$ expires (i.e., no request arises at $s_j$), if $s_{j}$ does not hold the only copy in the system, it drops the copy (lines 23-25). Otherwise, if $s_{j}$ holds the only copy, it continues to keep the copy by setting its tag $K_j$ (lines 21-22) until the next request in the system. If the next request arises at $s_j$, $s_j$ sets a new intended duration of the copy according to the prediction and clears its tag (lines 10-14). If the next request arises at another server $s_k$ ($s_k \neq s_j$), a transfer has to be performed to $s_k$ to serve the request. In this case, $s_j$ drops its copy right after the transfer and clears its tag (lines 15-19). Meanwhile, a copy will be created at $s_k$ upon receiving the transfer (lines 8-9), so the requirement of maintaining at least one copy is met.

\begin{figure}[htbp]
\centering
\includegraphics[width=9cm]{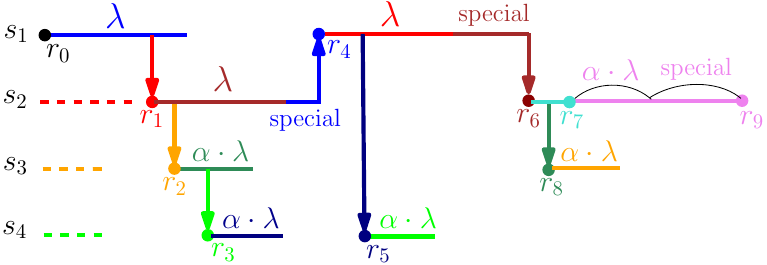}
\caption{An example of our online algorithm}
\label{fig:costalloc}
\end{figure}

To facilitate algorithm analysis, we refer to a data copy held within the intended duration derived from the prediction as a \textit{regular copy}, and a data copy held beyond the intended duration as a \textit{special copy}.
Figure \ref{fig:costalloc} shows an example produced by our algorithm, where horizontal edges represent data copies and vertical edges represent transfers.
By the algorithm definition, since a \textit{special copy} is the only copy in the system, it is easy to infer the following property.

\begin{Proposition}
The storage periods of any two special copies do not overlap. Moreover, the storage period of any special copy does not overlap with that of any regular copy.
\label{pro-3}
\end{Proposition}

\section{Preliminaries for Analysis}

Our general approach to competitive analysis is to first divide a request sequence into partitions based on the characteristics of an optimal offline strategy, and then study the costs of Algorithm \ref{alg3} and the optimal strategy for each partition separately. To this end, we need to allocate the cost produced by Algorithm \ref{alg3} (referred to as the {\em online cost} for short) to individual requests and characterize an optimal offline strategy to prepare for the analysis.

\subsection{Allocation of Online Cost}
\label{onlinealloc}
For each request $r_i$, we define $r_{p(i)}$ as the preceding request of $r_i$ arising at the same server,\footnote{For example, in Figure \ref{fig:costalloc}, $p(6) = 1$ since $r_1$
and $r_6$ arise in succession at server $s_2$.} and $l_i$ as the intended duration of the regular copy in server $s[r_i]$ after $r_{p(i)}$. By the algorithm definition, $l_i$ equals either $\hyperpar\cdot\trfrcost$ or $\trfrcost$. As illustrated in Figure \ref{type1234}, we categorize all the requests in the sequence into four types based on how they are served by our online algorithm.

\begin{figure}[htbp]
\centering
\includegraphics[width=12cm]{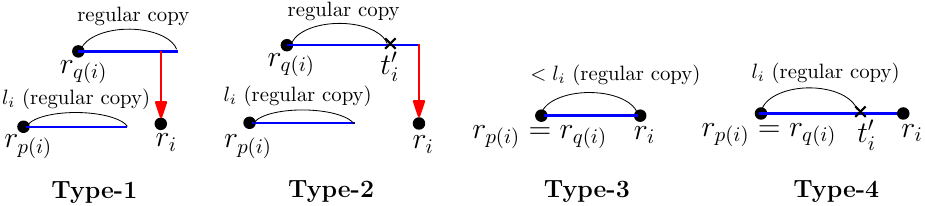}
\caption{Different request types in our online algorithm}
\label{type1234}
\end{figure}

If a request $r_i$ is served by a transfer from another server, by the algorithm definition, that server must keep a data copy since its most recent request before $r_i$. We denote this most recent request by $r_{q(i)}$ and its server by $s[r_{q(i)}]$. At the time of $r_i$, if the copy held by $s[r_{q(i)}]$ is a regular copy, $r_i$ is called a \textbf{Type-1} request. If the copy held by $s[r_{q(i)}]$ is a special copy, $r_i$ is called a \textbf{Type-2} request. For each \textbf{Type-2} request $r_{i}$, we define $t'_{i}$ as the time instant when the copy in $s[r_{q(i)}]$ switches from a regular copy to a special copy.

If a request $r_i$ is served by a local copy, by the algorithm definition, the copy must be kept in $s[r_i]$ since the preceding request $r_{p(i)}$ at $s[r_i]$. If the copy is a regular copy at $r_i$'s arrival, $r_i$ is called a \textbf{Type-3} request. If the copy is a special copy at $r_i$'s arrival, $r_i$ is called a \textbf{Type-4} request.

For notational convenience, for each \textbf{Type-3/4} request $r_{i}$, we define $r_{q(i)} = r_{p(i)}$. Then, $r_{q(i)}$ uniformly denotes the earlier request that provides the data copy to serve $r_i$ for all request types. For each \textbf{Type-4} request $r_i$, we also define $t'_{i}$ as the time instant when the copy in $s[r_{i}]$ switches from a regular copy to a special copy.

For each \textbf{Type-1/2} request $r_i$,
the cost allocated to $r_{i}$ includes: (1) the transfer cost of $\trfrcost$ for serving $r_{i}$; (2) the storage cost of the special copy in server $s[r_{q(i)}]$ over the period $(t'_{i},t_{i})$ if $r_{i}$ is a \textbf{Type-2} request; (3) the storage cost of the regular copy in server $s[r_{i}]$ after $r_{p(i)}$, which is given by $l_i$.\footnote{Note that it is possible for the copy in server $s[r_{i}]$ to switch from a regular copy to a special copy (i.e., extend beyond the $l_i$ period). In this case, the storage cost of the special copy is allocated to a request at another server served by a transfer from $s[r_{i}]$ according to the aforementioned (2).}

For each \textbf{Type-3/4} request $r_i$, we allocate the storage cost of the copy in server $s[r_{i}]$ during the period $(t_{p(i)},t_{i})$ to $r_i$, i.e., $r_i$ is allocated a cost of $t_{i}-t_{p(i)}$.

The cost allocation can be summarized as follows.
\begin{Proposition}
\label{costsummary}
The online cost allocated to a request $r_i$, based on its type, is given by
\begin{itemize}
\item \textbf{Type-1:}
$l_i+\trfrcost$;
\item \textbf{Type-2:}
$(t_{i}-t'_{i})+l_i+\trfrcost$;
\item \textbf{Type-3:}
$t_{i}-t_{p(i)}$;
\item \textbf{Type-4:}
$t_{i}-t_{p(i)}=(t_{i}-t'_{i})+l_i$.
\end{itemize}
\end{Proposition}

In the example of Figure \ref{fig:costalloc}, each request and its allocated cost are shown in the same color. $r_{5}$ and $r_{8}$ are \textbf{Type-1} requests, $r_{4}$ and $r_{6}$ are \textbf{Type-2} requests, $r_{7}$ is a \textbf{Type-3} request, and $r_{9}$ is a \textbf{Type-4} request.

There are some special considerations to note in the cost allocation. Given a request sequence $\langle r_{1}, r_{2},...,r_{m}\rangle$, after serving the final request $r_m$, a regular copy is created in server $s[r_m]$. Among this regular copy and all other regular copies in the system after $r_m$, the copy expiring the latest would switch to a special copy and stay infinitely.
In our analysis, we shall not account for the cost of the regular copy created in server $s[r_m]$ after $r_m$ and the special copy to stay infinitely. The rationale is that the storage periods of these two copies do not overlap and they are both beyond $r_m$. These two copies are considered to be in existence for ensuring that there is at least one copy in the system beyond $r_m$. In an optimal offline strategy, no copy will need to be stored beyond $r_m$. To focus on the cost in a finite time horizon, we do not account for the cost of the aforesaid two copies by the online algorithm.\footnote{If we insist in considering these two copies, their storage costs can be bounded by the storage cost of a data copy in any server beyond $r_m$ in the optimal offline strategy. Hence, it would not affect the correctness of our competitive analysis.}

By the above assumption, if there are $n$ servers ever holding copies in serving a request sequence, there will be $n-1$ regular copies after the last request at each server (except the server where $r_m$ arises), whose costs have not been allocated. Note that the first request at each server (except $s_{1}$ with the initial copy) must be served by a transfer because no copy was stored in the server. Since there are $n-1$ such first requests in total, we allocate the storage costs of the aforementioned $n-1$ regular copies to these $n-1$ first requests (one regular copy for each request). In this way, the first request $r_{j}$ of each server (except $s_1$) is allocated a total cost of $l+\trfrcost$ if it is served by a transfer from a regular copy (\textbf{Type-1} request), or a total cost of $(t_{j}-t'_{j}) + l + \trfrcost$ if it is served by a transfer from a special copy (\textbf{Type-2} request), where $l$ is the intended duration of the regular copy after the last request at some server). This is then consistent with the cost allocations of other \textbf{Type-1/2} requests as given in Proposition \ref{costsummary}.

In the example of Figure \ref{fig:costalloc}, $r_{1}$, $r_{2}$ and $r_{3}$ are the first requests arising at servers $s_{2}$, $s_{3}$ and $s_{4}$ respectively, and they are all \textbf{Type-1} requests. We do not account for the regular copy after the final request $r_9$ which arises at $s_2$. There are three regular copies after the last requests at the other servers (i.e., after $r_4$, $r_5$ and $r_8$). Their storage cost of $\trfrcost + 2\cdot\trfrcost\cdot\hyperpar$ is allocated to $r_{1}$, $r_{2}$ and $r_{3}$, with the cost of one regular copy for each of them (dashed lines).

It is easy to verify that the sum of the costs allocated to all requests is equal to the total online cost and the dummy request $r_0$ is not allocated any cost.

\subsection{Characteristics of An Optimal Strategy}
Deriving the exact form of an optimal offline strategy
for a request sequence is not straightforward. In this section, we present some characteristics of an optimal replication strategy to facilitate our competitive analysis.
The proofs of the following propositions are given in the appendix.

\begin{Proposition}
There exists an optimal replication strategy in which for each transfer, there is a request at either the source server or the destination server of the transfer.
\label{pro-1}
\end{Proposition}

The main idea to prove the above proposition is that if there is no request at the source and destination servers, we can always advance or delay the transfer to save or maintain the total cost.

\begin{Proposition}
There exists an optimal replication strategy with the characteristic in Proposition \ref{pro-1} and that for each request $r_i$,
if $r_i$ is served by a local copy, the copy is created no later than $r_{p(i)}$.
\label{pro-2}
\end{Proposition}

The following feature says that if two successive requests at the same server are sufficiently close in time, the server should hold a copy between them.

\begin{Proposition} \label{prop5}
There exists an optimal replication strategy with the characteristics in Propositions \ref{pro-1}, \ref{pro-2} and that for each request $r_i$, if $t_i - t_{p(i)} \leq \trfrcost$, server $s[r_i]$ holds a copy throughout the period $(t_{p(i)}, t_i)$, so that $r_i$ is served by a local copy.
\end{Proposition}

If a data copy is consistently stored in a server before and after a time instant $t$, we say that this copy {\em crosses} time $t$.

\begin{Proposition}
There exists an optimal replication strategy with the characteristics in Propositions \ref{pro-1}, \ref{pro-2}, \ref{prop5} and that for each request $r_i$, if $r_i$ is served by a transfer and no server holds a copy crossing the time $t_i$ of $r_i$, then (i) $r_{i-1}$ and $r_i$ arise at different servers; and (ii) the source server of the transfer is $s[r_{i-1}]$ which keeps a copy since $t_{i-1}$.
\label{prop-add}
\end{Proposition}

Hereafter, an optimal offline strategy shall always refer to one with the characteristics stated in Propositions \ref{pro-1}, \ref{pro-2}, \ref{prop5} and \ref{prop-add}.

\section{Division Approach for Analysis}
\label{sec:preparation}

We adopt a division approach to analyze the robustness and consistency of our algorithm.
Consider an optimal offline strategy of a request sequence $\left\langle r_{1},r_{2},\dots,r_{m} \right\rangle$. Remember that each request $r_i$ arises at server $s[r_i]$. For each request $r_i$, we examine if any server other than $s[r_i]$ holds a copy crossing the time $t_i$ of $r_i$.
We identify all such requests $r_i$'s where no other server holds a copy crossing $t_i$ and divide the request sequence into partitions at these requests. Note that the dummy request $r_0$ and the final request $r_m$ naturally meet the criterion.

Let
$\left\langle r_{d},r_{d+1},\dots,r_{e} \right\rangle$ denote a partition of requests, where $0 \leq d < e \leq m$.
That is, no copy at any other server crosses $t_d$, no copy at any other server crosses $t_e$, and for each request $r_i$ where $d < i < e$, there exists a copy at some other server crossing $t_i$.
We use $\textbf{Online}(d,e)$ to denote the total online cost allocated to the requests $r_{d+1},r_{d+2},\dots,r_{e}$ by the allocation method in Section \ref{onlinealloc}, and use $\textbf{OPT}(d,e)$ to denote the total cost of storage and transfers incurred over the period $(t_d, t_e]$ in the optimal offline strategy.
In the analysis of robustness (resp. consistency), we shall bound the ratio $\frac{\textbf{Online}(d,e)}{\textbf{OPT}(d,e)}$ by $1 + \frac{1}{\hyperpar}$ (resp. $\frac{5+\hyperpar}{3}$) for each partition. Then, aggregating over all the partitions, the ratio between the online cost and the optimal offline cost of the entire request sequence is bounded by the same constant.

To analyze $\textbf{OPT}(d,e)$, in the optimal offline strategy, we shall pick a set of storage periods of data copies to cover the time span from $t_d$ to $t_e$.

\medskip
\noindent \textbf{Case A: Request $r_e$ is served by a transfer.}
By Proposition \ref{prop-add}, in this case, $r_{i-1}$ and $r_{i}$ arise at different servers, and server $s[r_{e-1}]$ holds a data copy during the period $\left(t_{e-1}, t_{e}\right)$ and performs a transfer to serve $r_e$ at time $t_e$. It can further be proved that no server other than $s[r_{e-1}]$ can hold a copy crossing $t_{e-1}$. Assume on the contrary that there exists a server $s \neq s[r_{e-1}]$ that keeps a copy crossing $t_{e-1}$. This copy must be deleted at some time instant $t \leq t_{e}$, since there is no copy crossing $t_{e}$. Hence, this copy is not used for serving any local request at $s$ after $t_{e-1}$, since there is no request during $\left(t_{e-1}, t_{e}\right)$. Thus, the copy in $s$ can be removed during $\left(t_{e-1}, t\right)$. As a result, the total cost is reduced, contradicting the optimality of the offline strategy. Therefore, no copy at any other server crosses $t_{e-1}$. It follows from the partition definition that $d = e-1$. So, we just pick the storage period $(t_{e-1}, t_e)$ of the copy at $s[r_{e-1}]$ to cover the time span of the partition (see Figure \ref{cases}).

\medskip
\noindent \textbf{Case B: Request $r_e$ is served by a local copy.}
In this case, we pick a set of data copies iteratively in reverse order of time.
By Proposition \ref{pro-2}, the copy serving $r_e$ must be created no later than the time $t_{p(e)}$ of the preceding request $r_{p(e)}$.
First, we pick the storage period $(t_{p(e)}, t_e)$ of the copy at $s[r_e]$ and move backward to time $t_{p(e)}$.
By the partition definition, we must have $p(e) \geq d$ (otherwise, the copy crosses $t_d$).
If $p(e) = d$, we stop.
If $p(e) > d$, by definition, there exists a copy at some other server $s \neq s[r_{p(e)}]$ crossing $t_{p(e)}$.
It can be proved that the copy at $s$ must be kept till at least the first local request $r_g$ at $s$ after $t_{p(e)}$ and hence
will serve that request. If it does not serve any local request at $s$ (see Figure \ref{fig:cross} for an illustration), the copy can be deleted earlier at $t_{p(e)}$, because all transfers originating from the copy at $s$ after $t_{p(e)}$ can originate from the copy at $s[r_e]$ instead. This reduces the total cost and contradicts the optimality of the offline strategy.

\begin{figure}[htbp]
\centering
\includegraphics[width=9cm]{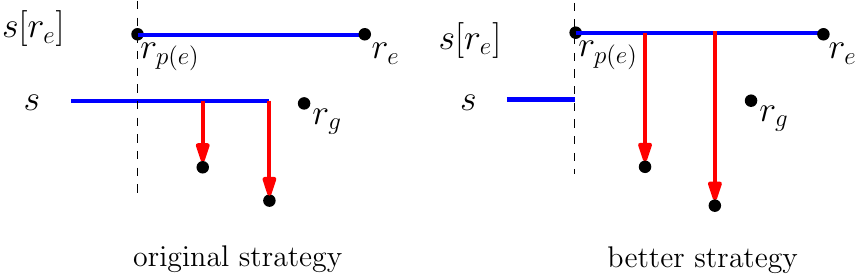}
\caption{Illustration of a data copy at $s$ crossing time $t_{p(e)}$}
\label{fig:cross}
\end{figure}

Since the copy at $s$ is kept to serve the first request $r_g$ at $s$ after $t_{p(e)}$, by Proposition \ref{pro-2}, the copy at $s$ must be created no later than the last request $r_{p(g)}$ at $s$ before $t_{p(e)}$.
So, we pick the storage period $(t_{p(g)}, t_g)$ of the copy and move backward to time $t_{p(g)}$.
If $p(g) = d$, we stop.
If $p(g) > d$, again there exists a copy at some other server $s' \neq s[r_{p(g)}]$ crossing $t_{p(g)}$.
Similarly, the copy must be kept till at least the first request $r_h$ at $s'$ after $t_{p(g)}$ and created no later than the last request $r_{p(h)}$ at $s'$ before $t_{p(g)}$.
Then, we pick the storage period $(t_{p(h)}, t_h)$ of the copy and move backward to time $t_{p(h)}$.
This backtracking process is repeated until time $t_d$ is reached.
Eventually, as illustrated in Figure \ref{cases}, we collect a set of storage periods of copies covering the time span of the partition. Each storage period has the form of $(t_{p(i)}, t_i)$, i.e., between two consecutive requests at the same server. We shall denote the set of storage periods picked as $Q = \{(t_{p({k_1})}, t_{k_1}), (t_{p({k_2})}, t_{k_2}), \dots, (t_{p({k_{|Q|}})}, t_{k_{|Q|}})\}$, where $d = p(k_1) < k_1 < k_2 < \cdots < k_{|Q|} = e$, each  $(t_{p(k_j)}, t_{k_j})$ crosses the start time $t_{p(k_{j+1})}$ of the next period, and $|Q|$ is the number of storage periods picked.

\begin{figure}[t]
\centering
\includegraphics[width=8.5cm]{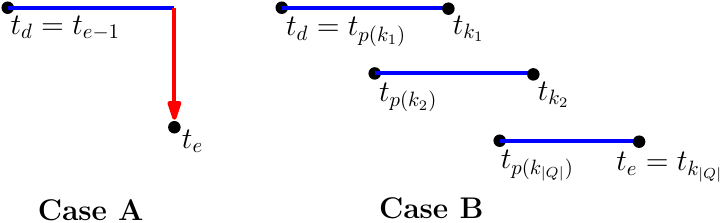}
\caption{Illustration of \textbf{Case A} and \textbf{Case B}}
\label{cases}
\end{figure}

\section{Robustness Analysis}
\label{sec:robustness}

In this section, we show that our proposed algorithm is $(1+\frac{1}{\hyperpar})$-robust, i.e., the competitive ratio of the algorithm is always bounded by $1+\frac{1}{\hyperpar}$, irrespective of the prediction correctness. We analyze partitions of cases A and B separately.

\medskip
\noindent \textbf{Case A:} In this case, $d = e-1$, i.e., the partition has the form of $\left\langle r_{e-1},r_{e} \right\rangle$. So, $\textbf{Online}(d,e)$ refers to the online cost allocated to request $r_{e}$, and $\textbf{OPT}(d,e)$ refers to the cost incurred over the period $(t_{e-1}, t_e]$ in the optimal offline strategy. By the analysis in Section \ref{sec:preparation}, in the optimal offline strategy, server $s[r_{e-1}]$ holds the only copy during $(t_{e-1}, t_e)$ and performs an outgoing transfer to serve $r_e$. Hence, $\textbf{OPT}(d,e) = (t_e - t_{e-1}) + \trfrcost$.

According to Proposition \ref{costsummary}, if $r_e$ is a \textbf{Type-1} request by our online algorithm, the online cost allocated to $r_e$ is $l_e + \trfrcost$ (where $l_e$ is the duration of the regular copy after $r_{p(e)}$). Recall that $l_e$ equals either $\hyperpar\cdot\trfrcost$ or $\trfrcost$. Since $l_e \leq \trfrcost$, we have $\frac{\textbf{Online}(d,e)}{\textbf{OPT}(d,e)} \leq \frac{\trfrcost + \trfrcost}{\trfrcost} =2 \leq 1 + \frac{1}{\hyperpar}$.

If $r_e$ is a \textbf{Type-2/4} request, the online cost allocated to $r_e$ is at most $\left(t_{e}-t'_{e}\right)+l_e+\trfrcost$. By Proposition \ref{pro-3}, there is no overlap between the storage period of the regular copy after $r_{e-1}$ and the storage period $\left(t'_{e},t_{e}\right)$ of the special copy for serving $r_e$. This implies that $\left(t'_{e},t_{e}\right)$ is completely contained in $\left(t_{e-1},t_{e}\right)$, so $t_{e}-t'_{e}<t_{e}-t_{e-1}$. Thus, $\frac{\textbf{Online}(d,e)}{\textbf{OPT}(d,e)} < \frac{\left(t_{e}-t_{e-1}\right)+l_e+\trfrcost}{(t_e - t_{e-1}) + \trfrcost} < 1+\frac{\trfrcost}{\trfrcost} =2 \leq 1 + \frac{1}{\hyperpar}$.

If $r_e$ is a \textbf{Type-3} request, the online cost allocated to $r_e$ is $t_{e}-t_{p(e)}$. By definition, $t_{e}-t_{p(e)} \leq l_e \leq \trfrcost$. Hence, $\frac{\textbf{Online}(d,e)}{\textbf{OPT}(d,e)} \leq \frac{\trfrcost}{\trfrcost} = 1 < 1 + \frac{1}{\hyperpar}$.

\medskip
\noindent \textbf{Case B:} In this case, we have a set of storage periods $Q$ covering the time span of the partition. For each storage period $(t_{p({k_j})}, t_{k_j}) \in Q$, we refer to the request $r_{k_j}$ at time $t_{k_j}$ as the {\em end sentinel request}. To calculate the optimal offline cost, we divide all requests $r_{d+1},r_{d+2},\dots,r_{e}$ into: $R_{Q} = \{r_{k_1}, r_{k_2}, \dots, r_{k_{|Q|}}\}$ (the end sentinel requests), $R_L$ ($r_i$'s that are not end sentinel requests and $t_i - t_{p(i)} \leq \trfrcost$), and $R_T$ ($r_i$'s that are not end sentinel requests and $t_i - t_{p(i)} > \trfrcost$). By Proposition \ref{prop5}, each request $r_i \in R_L$ is served by a local copy over the period $(t_{p(i)}, t_i)$. Since the storage periods $Q$ cover the time span, to minimize cost, each request $r_i \in R_T$ must be served by a transfer from a copy of some storage period in $Q$.
Thus, the optimal offline cost can be written as
\begin{equation}
\textbf{OPT}(d,e) = \sum_{r_i \in R_{Q}}(t_i-t_{p(i)}) + \sum_{r_i \in R_L}(t_i-t_{p(i)}) + \sum_{r_i \in R_T} \trfrcost.
\label{opt}
\end{equation}

On the other hand, by our online algorithm, we divide all requests $r_{d+1},r_{d+2},\dots,r_{e}$ into: $R_1$ (\textbf{Type-1} requests), $R_2$ (\textbf{Type-2} requests), $R_3$ (\textbf{Type-3} requests), and $R_4$ (\textbf{Type-4} requests). By Proposition \ref{costsummary}, the total online cost allocated to all requests can be written as
\begin{equation}
\textbf{Online}(d,e)
=
\sum_{r_i \in R_1 \cup R_2}
\trfrcost + \sum_{r_i \in R_1 \cup R_2 \cup R_4}
l_i + \sum_{r_i \in R_2 \cup R_4}
\!\!\!
\left(t_i-t'_i\right) + \sum_{r_i \in R_3}(t_i-t_{p(i)}).
\label{onlinecost}
\end{equation}

By the algorithm definition, special copies cannot cross the time of any request, because each special copy is deleted after serving the next request. Recall that $Q = \{(t_{p({k_1})}, t_{k_1}), (t_{p({k_2})}, t_{k_2}), \dots,$ $(t_{p({k_{|Q|}})}, t_{k_{|Q|}})\}$, where $p(k_1) = d$, $k_{|Q|} = e$, and each  $(t_{p(k_j)}, t_{k_j})$ crosses time $t_{p(k_{j+1})}$. Thus, we can divide the time span of the partition into the following disjoint intervals: $(t_{p({k_1})}, t_{p(k_2)})$, $(t_{p({k_2})},$ $t_{p(k_3)})$, \dots, $(t_{p({k_{|Q|-1}})}, t_{p(k_{|Q|})})$ and $(t_{p({k_{|Q|}})}, t_{k_{|Q|}})$.
The storage period $\left(t'_i, t_i\right)$ of the special copy for serving each request $r_i \in R_2 \cup R_4$ must be completely contained in one of these intervals, because there are requests at both endpoints of these intervals. For those intervals containing special copies, we select their original periods in $Q$ to form a subset $X \subset Q$, i.e., for each $j < |Q|$, $(t_{p({k_j})}, t_{k_j}) \in X$ if there exist special copies during $(t_{p({k_j})}, t_{p(k_{j+1})})$; and $(t_{p({k_{|Q|}})},$ $t_{k_{|Q|}}) \in X$ if there exist special copies during $(t_{p({k_{|Q|}})}, t_{k_{|Q|}})$. Let $R_X$ denote the end sentinel requests of all storage periods in $X$. Each period $(t_{p(i)}, t_{i})\in X$ covers all special copies therein and $l_{i}$. By Proposition \ref{pro-3}, all the special copies do not overlap. Thus, we have
\begin{equation}
\sum_{(t_{p(i)},t_i)\in X} \!\!\! (t_{i} - t_{p(i)}) = \!\!\! \sum_{r_i \in R_{X}}(t_i-t_{p(i)}) \geq \!\!\! \sum_{r_i \in R_2 \cup R_4} \!\!\! \left(t_i-t'_i\right)+\sum_{r_i \in R_{X}}l_i
\label{upperspecial0}
\end{equation}

In addition, by our online algorithm, a special copy occurs only when all the regular copies expire. Thus, for each end sentinel request $r_i \in R_X$, the regular copy after the preceding request  $r_{p(i)}$ must expire before $r_i$'s arrival. This implies that $r_i$ cannot be a \textbf{Type-3} request.

\begin{Proposition}
All requests in $R_X$ are
\textbf{Type-1/2/4} requests, i.e., $R_{X} \subset R_1 \cup R_2 \cup R_4$.
\label{prop:rx}
\end{Proposition}

Hence, applying  (\ref{upperspecial0}) to (\ref{onlinecost}), we can bound the total online cost by
\begin{equation}
\textbf{Online}(d,e)\leq \sum_{r_i \in R_1 \cup R_2} \trfrcost + \sum_{r_i \in R_1 \cup R_2 \cup R_4 \setminus R_{X}} l_i +\sum_{r_i \in R_3\cup R_X}(t_i-t_{p(i)}).
\label{upperonline}
\end{equation}

We show $\big(1+\frac{1}{\hyperpar}\big)$-robustness by comparing the corresponding terms in (\ref{opt}) and (\ref{upperonline}) for individual requests. Observe that for each request $r_i \notin R_3$, we have $t_i - t_{p(i)} \geq l_i \geq \hyperpar \cdot \trfrcost$ since the regular copy after $r_{p(i)}$ must expire before $r_i$'s arrival and a regular copy is at least $\hyperpar \cdot \trfrcost$ long. On the other hand, for each request $r_i \in R_3$, we have $t_i - t_{p(i)} \leq l_i \leq \trfrcost$ since a regular copy is at most $\trfrcost$ long.

By definition, $R_{X} \subset R_{Q}$. For each request $r_i \in R_{X}$, since $r_i \notin R_3$ by Proposition \ref{prop:rx}, we have $t_i - t_{p(i)} \geq \hyperpar \cdot \trfrcost$. The online cost for $r_i$ in (\ref{upperonline}) is at most $\trfrcost + t_i - t_{p(i)}$, while the offline cost for $r_i$ in (\ref{opt}) is $t_i - t_{p(i)}$. Hence, the online-to-offline ratio for $r_i$ is bounded by $\frac{\trfrcost + t_i - t_{p(i)}}{t_i - t_{p(i)}} \leq 1 + \frac{\trfrcost}{t_i - t_{p(i)}} \leq 1 + \frac{1}{\hyperpar}$.

For each request $r_i \in R_{Q} \setminus R_{X}$ or $r_i \in R_{L}$, the offline cost for $r_i$ in (\ref{opt}) is $t_i - t_{p(i)}$. If $r_i \in R_3$, the online cost for $r_i$ in (\ref{upperonline}) is the same. If $r_i \notin R_3$, the online cost for $r_i$ in (\ref{upperonline}) is at most $\trfrcost + l_i$. Since $t_i - t_{p(i)} \geq l_i \geq \hyperpar \cdot \trfrcost$, the online-to-offline ratio for $r_i$ is bounded by $\frac{\trfrcost+l_i}{t_i - t_{p(i)}} \leq \frac{\trfrcost+l_i}{l_i}\leq 1 + \frac{1}{\hyperpar}$.

For each request $r_i \in R_{T}$, the offline cost for $r_i$ in (\ref{opt}) is $\trfrcost$. If $r_i \in R_3$, the online cost for $r_i$ in (\ref{upperonline}) is $t_i - t_{p(i)} \leq \trfrcost$. Thus, the online-to-offline ratio for $r_i$ is bounded by $1$. If $r_i \notin R_3$, the online cost for $r_i$ in (\ref{upperonline}) is at most $\trfrcost + l_i$. Since $l_i \leq \trfrcost$, the online-to-offline ratio for $r_i$ is bounded by $\frac{\trfrcost+l_i}{\trfrcost} \leq 2\leq 1 + \frac{1}{\hyperpar}$.

Since the online-to-offline ratio for each individual request in the partition is bounded by $1 + \frac{1}{\hyperpar}$, it can be concluded that $\frac{\textbf{Online}(d,e)}{\textbf{OPT}(d,e)} \leq
1+\frac{1}{\hyperpar}$.

\begin{figure}[htbp]
\centering
\includegraphics[width=7.5cm]{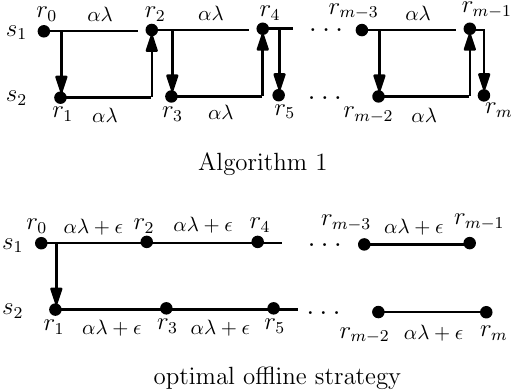}
\caption{A tight example for robustness analysis}
\label{robustness}
\end{figure}

Figure \ref{robustness} gives an example to show that the robustness analysis above is tight. A total of $m+1$ requests arise at two servers $s_1$ and $s_2$. The dummy request $r_0$ arises at $s_1$ at time $t_0 = 0$. The first request $r_1$ arises at $s_2$ at time $t_1 = \epsilon$. The inter-request time between two consecutive requests at $s_1$ or $s_2$ is $\hyperpar\trfrcost+\epsilon$, where $\epsilon>0$ is an small value. Suppose that the inter-request times at $s_1$ and $s_2$ are always predicted to be greater than $\trfrcost$. Then, the intended duration of the regular copy after each request is $\hyperpar \trfrcost$. By our online algorithm, $r_1$ is served by a transfer. $r_2$ arises right after the expiration of the regular copy in $s_1$, so it is served by a transfer from $s_2$. After the outgoing transfer, the regular copy in $s_2$ expires, and then $r_3$ arises at $s_2$ shortly afterwards. Hence, $r_3$ is served by a transfer from $s_1$. This pattern is repeated continuously, so that all the requests $r_2, r_3, \dots, r_{m}$ are served by transfers, as they all arise after the expiration of the preceding regular copies.  Thus, the total online cost is $(m-1)\cdot (\alpha\trfrcost+\trfrcost)+\trfrcost$.
In the optimal offline strategy, after $r_1$, both $s_1$ and $s_2$ keep copies to serve all their local requests. So, the optimal cost is $(m-1)\cdot(\hyperpar \trfrcost+\epsilon)+ \trfrcost$.
Therefore, the online-to-optimal cost ratio is $\frac{(m-1)\cdot (\alpha\trfrcost+\trfrcost)+\trfrcost}{(m-1)\cdot (\alpha\trfrcost+\epsilon)+\trfrcost}$, which approaches $1+\frac{1}{\alpha}$ as $m\rightarrow \infty$ and $\epsilon\rightarrow 0$.

\section{Consistency Analysis}
\label{sec:consistency}

In this section, we show that our proposed algorithm is $\frac{5+\hyperpar}{3}$-consistent, i.e., the competitive ratio of the algorithm is $\frac{5+\hyperpar}{3}$ when all predictions are correct. We start with an important observation for correct predictions.

\begin{Proposition}
Given that the prediction
of the inter-request time between $r_{p(i)}$ and $r_i$
is correct, if $r_i$ is a \textbf{Type-1/2/4} request, then $t_{i}-t_{p(i)} > \trfrcost$ and $l_i = \hyperpar\cdot\trfrcost$; if $r_i$ is a \textbf{Type-3} request, then $t_{i}-t_{p(i)} \leq \trfrcost$ and $l_i = \trfrcost$.
\label{consistrequest}
\end{Proposition}
\begin{proof}
If $r_{i}$ is a \textbf{Type-1/2/4} request, by definition, the inter-request time between $r_{p(i)}$ and $r_{i}$ (i.e., $t_{i}-t_{p(i)}$) must be longer than the duration $l_i$ of the regular copy after $r_{p(i)}$. Since the prediction is correct, if $t_{i}-t_{p(i)} \leq \trfrcost$, then $l_i=\trfrcost \geq t_{i}-t_{p(i)}$, leading to a contradiction. Thus, $t_{i}-t_{p(i)} > \trfrcost$ and hence, $l_i=\hyperpar\cdot\trfrcost$. If $r_{i}$ is a \textbf{Type-3} request, by definition, the inter-request time $t_{i}-t_{p(i)}$ must be no longer than $l_i$. Since the prediction is correct, if $t_{i}-t_{p(i)} > \trfrcost$, then $l_i= \hyperpar\cdot\trfrcost \leq \trfrcost < t_{i}-t_{p(i)}$, leading to a contradiction. Thus, $t_{i}-t_{p(i)} \leq \trfrcost$ and hence, $l_i=\trfrcost$.
\end{proof}

Again, we analyze partitions of cases A and B (as presented in Section \ref{sec:preparation}) separately.

\medskip
\noindent \textbf{Case A:} In this case, the partition has the form of $\left\langle r_{e-1},r_{e} \right\rangle$, i.e., $d=e-1$. Same as in the robustness analysis, $\textbf{OPT}(d,e) = (t_e - t_{e-1}) + \trfrcost$, and $\textbf{Online}(d,e)$ is the online cost allocated to $r_{e}$. Since $r_e$ is served by a transfer in the optimal offline strategy, by Proposition \ref{prop5}, $t_{e}-t_{p(e)}>\trfrcost$. It follows from Proposition \ref{consistrequest} that $r_e$ is a \textbf{Type-1/2/4} request by our online algorithm and $l_{e}=\hyperpar\trfrcost$.

Based on Proposition \ref{costsummary}, if $r_{e}$ is a \textbf{Type-1} request, the online cost allocated to $r_e$ is $l_e+\trfrcost=\hyperpar\trfrcost+\trfrcost$. Thus, $\frac{\textbf{Online}(d,e)}{\textbf{OPT}(d,e)} = \frac{\hyperpar\trfrcost+\trfrcost}{(t_e - t_{e-1}) + \trfrcost} < 1+ \hyperpar \leq \frac{5+\hyperpar}{3}$.

If $r_e$ is a \textbf{Type-2/4} request, the online cost allocated to $r_e$ is at most $\left(t_{e}-t'_{e}\right)+l_e+\trfrcost=\left(t_{e}-t'_{e}\right)+\hyperpar\trfrcost+\trfrcost$. By Proposition \ref{pro-3}, there is no overlap between the storage period of the regular copy after $r_{e-1}$ and the storage period $\left(t'_{e},t_{e}\right)$ of the special copy for serving $r_e$. Hence, $\left(t'_{e},t_{e}\right)$ is fully contained in $\left(t_{e-1},t_{e}\right)$ so that $t_{e}-t'_{e}<t_{e}-t_{e-1}$. Thus, $\frac{\textbf{Online}(d,e)}{\textbf{OPT}(d,e)} < \frac{\left(t_{e}-t_{e-1}\right)+\hyperpar\trfrcost+\trfrcost}{(t_e - t_{e-1}) + \trfrcost} < 1 + \hyperpar \leq \frac{5+\hyperpar}{3}$.

\medskip
\noindent\textbf{Case B:} In this case, we have picked a set of storage periods $Q$ covering the time span of the partition. Recall that $R_1/R_2/R_3/R_4$ denotes all \textbf{Type-1/2/3/4} requests among $r_{d+1},r_{d+2},\dots,r_{e}$ by our online algorithm. By Proposition \ref{consistrequest}, for each request $r_i \in R_1\cup R_2\cup R_4$, $l_i=\hyperpar\trfrcost$. Hence, the online cost (\ref{onlinecost}) can be rewritten as
\begin{equation}
\textbf{Online}(d,e) =
\left(|R_1| + |R_2|\right) \cdot \left(1 + \hyperpar\right)\trfrcost + |R_4| \cdot \hyperpar \trfrcost
+ \sum_{r_i \in R_2\cup R_4} (t_i-t'_i) + \sum_{r_i \in R_{3}}(t_i - t_{p(i)}).
\label{online}
\end{equation}

Recall that $X$ is the set of storage periods in $Q$ that contain special copies. Using the definition of $X$, we can bound the total cost of special copies in (\ref{online}).
The proofs of the lemmas hereafter are given in the appendix.

\begin{Lemma}
When all predictions are correct, we have 
\begin{equation*}
    \sum_{r_i\in R_2\cup R_4}\left(t_{i}-t'_{i}\right) \leq \sum_{r_i\in R_X}(t_{i} - t_{p(i)})-\left(|R_2|+|R_4|\right)\cdot\hyperpar\trfrcost.
\end{equation*}
\label{lem-2}
\end{Lemma}

By (\ref{online}) and Lemma \ref{lem-2}, we have
\begin{equation}
\textbf{Online}(d,e) \leq |R_1| \cdot \left(1 + \hyperpar\right)\trfrcost + |R_2| \cdot \trfrcost+\sum_{r_i\in R_{3}\cup R_X}(t_i-t_{p(i)}).
\label{upboundonline}
\end{equation}

Recall that in the optimal offline strategy, we divide all requests $r_{d+1},r_{d+2},\dots,r_{e}$ into: $R_{Q}$ (the end sentinel requests), $R_L$ ($r_i$'s that are not end sentinel requests and $t_i - t_{p(i)} \leq \trfrcost$), and $R_T$ ($r_i$'s that are not end sentinel requests and $t_i - t_{p(i)} > \trfrcost$). According to the request types by our online algorithm, we further divide the end sentinel requests $R_Q$ into: $R'_1$ (\textbf{Type-1} requests), $R'_2$ (\textbf{Type-2} requests), $R'_3$ (\textbf{Type-3} requests), and $R'_4$ (\textbf{Type-4} requests). Obviously, it holds that
\begin{equation}
|Q|=|R_{Q}| = |R'_1|+|R'_2|+|R'_3|+|R'_4|.
\label{Qdef}
\end{equation}
Moreover, it follows from $R_X \subset R_Q$ and Proposition \ref{prop:rx} that $R_X \subset  R'_1\cup R'_2 \cup R'_4$.
In addition, by Proposition \ref{consistrequest}, all  requests in $R_L$ are \textbf{Type-3} requests and all requests in $R_T$ are \textbf{Type-1/2/4} requests. Thus, $R_L = R_3\setminus R'_3$ and $R_T= (R_1\cup R_2\cup R_4)\setminus(R'_1\cup R'_2\cup R'_4)$.
Hence, the optimal offline cost (\ref{opt}) can be rewritten as
\begin{eqnarray}
\textbf{OPT}(d,e) & = &\sum_{r_i \in R'_3 \cup R_{X}}(t_i-t_{p(i)}) + \sum_{r_i \in R'_1\cup R'_2\cup R'_4 \setminus R_{X}}(t_i-t_{p(i)}) 
\nonumber \\
& & + \sum_{r_i \in R_3\setminus R'_3}(t_i-t_{p(i)}) + \sum_{r_i \in (R_1\cup R_2\cup R_4)\setminus(R'_1\cup R'_2\cup R'_4)} \trfrcost \nonumber \\
& \geq &\sum_{r_i \in R_3 \cup R_{X}}(t_i-t_{p(i)}) + \sum_{r_i \in R'_1\cup R'_2\cup R'_4 \setminus R_{X}} \trfrcost
\nonumber \\
& & + \sum_{r_i \in (R_1\cup R_2\cup R_4)\setminus(R'_1\cup R'_2\cup R'_4)} \trfrcost \quad (\text{by Proposition \ref{consistrequest}}) \nonumber \\
& = & (|R_{1}| + |R_{2}| + |R_{4}| - |X|) \cdot \trfrcost + \sum_{r_i \in R_3 \cup R_{X}}(t_i-t_{p(i)}).\quad\,\,
\label{RRR}
\end{eqnarray}

Combining (\ref{upboundonline}) and (\ref{RRR}), we have
\begin{eqnarray}
\frac{\textbf{Online}(d,e)}{\textbf{OPT}(d,e)} 
& \leq & 1+\frac{|R_1| \cdot \hyperpar\trfrcost - |R_4| \cdot \trfrcost + |X| \cdot \trfrcost}{(|R_{1}| + |R_{2}| + |R_{4}| - |X|) \cdot \trfrcost + \sum\limits_{r_i\in R_{3}\cup R_X}(t_i - t_{p(i)})}
\nonumber \\
& \leq & 1+\frac{|R_1| \cdot \hyperpar\trfrcost - |R_4| \cdot \trfrcost + |X| \cdot \trfrcost}{(|R_{1}| + |R_{2}| + |R_{4}|) \cdot \trfrcost}  \quad (\text{by Propositions \ref{prop:rx}, \ref{consistrequest}}) \nonumber \\
& = & 1+\hyperpar+\frac{|X|-\left(|R_2| + |R_4|\right) \cdot \hyperpar - |R_4|}{|R_{1}| + |R_{2}| + |R_{4}|}.
\label{upperratio}
\end{eqnarray}

By definition, for each $(t_{p({k_j})}, t_{k_j}) \in X$, if $j < |Q|$, the period $(t_{p({k_j})}, t_{p(k_{j+1})})$ contains at least one special copy and hence one \textbf{Type-2/4} request; if $j = |Q|$, the period $(t_{p({k_{|Q|}})}, t_{k_{|Q|}})$ contains at least one special copy and hence one \textbf{Type-2/4} request. Since these periods are disjoint, all these \textbf{Type-2/4} requests are distinct. Thus, $|X|\leq |R_2|+|R_4|$. As a result,
\begin{equation}
\frac{\textbf{Online}(d,e)}{\textbf{OPT}(d,e)} \leq 1+\hyperpar+\frac{|X| \cdot(1-\hyperpar)-|R_4|}{|R_{1}| + |R_{2}| + |R_{4}|}.
\label{ratiobound}
\end{equation}

In order to bound the above ratio, we discuss the cases of $|Q| \geq 2$ and $|Q| = 1$ separately.

\medskip
\noindent\textbf{Case 1: $|Q| \geq 2$.} In this case, we have the following result about the total number of \textbf{Type-1/2/4} requests in the partition.

\begin{Lemma} When $|Q| \geq 2$ and $|X|\geq 1$, it holds that
$|R_{1}|+|R_{2}|+|R_{4}|\geq \max\{2|X|-1,2\}$.
\label{lem-5}
\end{Lemma}

By (\ref{ratiobound}) and Lemma \ref{lem-5}, we have:

$\bullet$ When $|X|\geq 2$,  $\frac{\textbf{Online}(d,e)}{\textbf{OPT}(d,e)}\leq 1+\hyperpar+\frac{|X|\cdot\left(1-\hyperpar\right)}{2|X|-1} \leq 1+\hyperpar+\frac{2\cdot\left(1-\hyperpar\right)}{3}=\frac{5+\hyperpar}{3}$.

$\bullet$ When $|X|=1$, $\frac{\textbf{Online}(d,e)}{\textbf{OPT}(d,e)}\leq 1+\hyperpar+\frac{\left(1-\hyperpar\right)}{2}
\leq\frac{5+\hyperpar}{3}$.

$\bullet$ When $|X|=0$,  $\frac{\textbf{Online}(d,e)}{\textbf{OPT}(d,e)} = 1+\hyperpar \leq\frac{5+\hyperpar}{3}$.

\medskip
\noindent\textbf{Case 2: $|Q|=1$.} In this case, there is only one storage period $(t_{p(k_1)},t_{k_1})\in Q$
and $|X|\leq |Q|=1$.

If the end sentinel request $r_{k_1}$ is a \textbf{Type-3} request by our online algorithm, the copy during $(t_{p(k_1)},t_{k_1})$ is a regular copy. By Proposition \ref{pro-3}, there is no special copy during $(t_{p(k_1)},t_{k_1})$ and hence $|X| = 0$. It follows from (\ref{ratiobound}) that $\frac{\textbf{Online}(d,e)}{\textbf{OPT}(d,e)} \leq 1+\hyperpar \leq \frac{5+\hyperpar}{3}$. If $r_{k_1}$ is a \textbf{Type-4} request, we have $|R_4|\geq 1$ and thus $\frac{\textbf{Online}(d,e)}{\textbf{OPT}(d,e)}\leq 1+\hyperpar+\frac{\left(1-\hyperpar\right)-1}{1} = 1 < \frac{5+\hyperpar}{3}$. If $r_{k_1}$ is a \textbf{Type-1/2} request,  we have the following lemma.

\begin{Lemma}
When $|Q|=1$, if the only end sentinel request $r_{k_1}$ is a \textbf{Type-1/2} request, there are at least two \textbf{Type-1/2/4} requests in the partition, i.e., $|R_{1}| + |R_{2}| + |R_{4}|\geq 2$.
\label{Lem-b}
\end{Lemma}

By (\ref{ratiobound}) and Lemma \ref{Lem-b}, we have $\frac{\textbf{Online}(d,e)}{\textbf{OPT}(d,e)}\leq 1+\hyperpar+\frac{1-\hyperpar}{2}
\leq \frac{5+\hyperpar}{3}$ if $r_{k_1}$ is a \textbf{Type-1/2} request.

In conclusion, for both cases of $|Q| \geq 2$ and $|Q| = 1$, the ratio $\frac{\textbf{Online}(d,e)}{\textbf{OPT}(d,e)}$ is bounded by $\frac{5+\hyperpar}{3}$.

\begin{figure}[htbp]
\centering
\includegraphics[width=10cm]{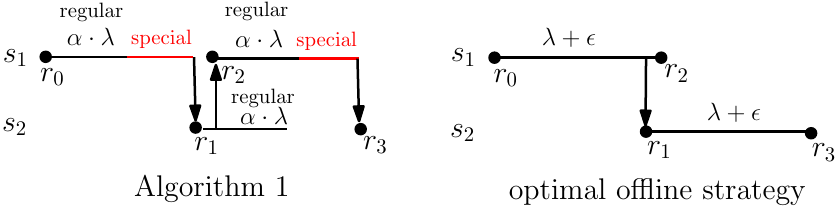}
\caption{A tight example for consistency analysis}
\label{AAAA7}
\end{figure}

Figure \ref{AAAA7} gives an example to show that the consistency analysis above is tight. Request $r_2$ arises at server $s_1$, and requests $r_1$ and $r_3$ arise at server $s_2$. The arising times of the requests are $t_1 = \trfrcost$, $t_2 = \trfrcost+ \epsilon$ and $t_3 = 2 \trfrcost + \epsilon$, where $\epsilon > 0$ is a small value. In the optimal offline strategy, both $r_2$ and $r_3$ are served locally, and $r_1$ is served by a transfer from $s_1$. So, the optimal cost is $3\trfrcost + 2\epsilon$. By our online algorithm, since the inter-request time at each server is longer than $\trfrcost$, the duration of each regular copy is $\hyperpar\trfrcost$ if all predictions are correct. When the regular copy after the dummy request $r_0$ expires, it switches to a special copy. Upon an outgoing transfer to serve $r_1$, the special copy is deleted. Then, $r_2$ has to be served by a transfer from the regular copy after $r_1$. Upon expiration, the regular copy after $r_2$ switches to a special copy since it is the only copy in the system. This special copy is later used to serve $r_3$ by a transfer. The total online cost is $5\trfrcost+\hyperpar\trfrcost$. Hence, the online-to-optimal cost ratio is $\frac{5\trfrcost+\hyperpar\trfrcost}{3\trfrcost + 2\epsilon}$, which approaches $\frac{5+\hyperpar}{3}$ as $\epsilon \rightarrow 0$. We note that this example can also be extended to a request sequence of arbitrary length by repeating $r_1$, $r_2$, $r_3$ (i.e., treating $r_3$ as $r_0$ of a new cycle with the roles of servers $s_1$ and $s_2$ swapped).

\section{Impact of Mispredictions}
\label{misprediction}

In this section, we study the impact of mispredictions on the competitive ratio of the proposed algorithm. For each request $r_i$, if the inter-request time between $r_{p(i)}$ and $r_{i}$ is mispredicted, we call $r_i$ a \textit{mispredicted request}. Otherwise, we call $r_i$ a \textit{correct request}. For each mispredicted request $r_i$, we investigate how $r_i$ may be affected and how correct requests may be affected.
Recall from Proposition \ref{costsummary} that the online cost allocated to a request consists of three possible parts: (1) the cost of a regular copy; (2) the cost of a special copy (if any); (3) the cost of a transfer (if any). We study the impact on these three parts.

We divide all the mispredicted requests into three sets $M_1$, $M_2$ and $M_3$ based on the real inter-request time between the mispredicted request and its preceding request.

For each $r_i\in M_1$, the inter-request time satisfies $t_i-t_{p(i)}\leq\hyperpar\trfrcost$. In the case of correct prediction, the duration $l_i$ of the regular copy after $r_{p(i)}$ is set to $\trfrcost$, so $r_i$ is a \textbf{Type-3} request since $t_i-t_{p(i)} \leq$ $\trfrcost$. In the case of misprediction, $l_i$ reduces to $\hyperpar\trfrcost$, and $r_i$ remains a \textbf{Type-3} request. Therefore, misprediction of the requests in $M_1$ does not affect the online cost.

For each $r_i\in M_2$, the inter-request time satisfies $\hyperpar\trfrcost<t_i-t_{p(i)}\leq\trfrcost$. In the case of correct prediction, $l_i$ is set to $\trfrcost$, so $r_i$ is a \textbf{Type-3} request since $t_i-t_{p(i)}\leq \trfrcost$. In the case of misprediction, $l_i$ reduces to $\hyperpar\trfrcost$, and $r_i$ becomes a \textbf{Type-1/2/4} request.
\begin{itemize}
\item By Proposition \ref{costsummary}, the transfer cost allocated to $r_i$ can increase by at most $\trfrcost$. The regular copy after $r_{p(i)}$ is shortened from $t_i-t_{p(i)}$ to $\hyperpar\trfrcost$, so the regular copy cost allocated to $r_i$ is reduced by $t_i-t_{p(i)} - \hyperpar\trfrcost$. Meanwhile, the cutback of the regular copy after $r_{p(i)}$ can create special copies at server $s[r_i]$ and other servers. By Proposition \ref{pro-3}, the total duration of all special copies created must be bounded by the duration cut of the regular copy after $r_{p(i)}$, i.e., $t_i-t_{p(i)} - \hyperpar\trfrcost$. The costs of the special copies created may be allocated to $r_i$ or correct requests, which would be offset by the aforesaid reduction in the regular copy cost allocated to $r_i$.
\item For correct requests, a \textbf{Type-1} correct request may change to \textbf{Type-2/4}. In this case, the transfer cost allocated to the request can never increase. A \textbf{Type-2} (resp. \textbf{Type-4}) correct request remains \textbf{Type-2} (resp. \textbf{Type-4}), since the special copy for serving the request is not affected by the cutback of the regular copy after $r_{p(i)}$. A \textbf{Type-3} correct request remains \textbf{Type-3}. In these cases, the transfer cost allocated to the request does not change.
In addition, the regular copy costs allocated to all correct requests do not change.
\item In summary, the total increase in the online cost due to each $r_i \in M_2$ is bounded by $\trfrcost$.
\end{itemize}

For each $r_i\in M_3$, the inter-request time $t_i-t_{p(i)}>\trfrcost$. In the case of correct prediction, $l_i$ is set to $\hyperpar\trfrcost$, so $r_i$ is a \textbf{Type-1/2/4} request since $t_i-t_{p(i)}>l_i$. In the case of misprediction, $l_i$ increases to $\trfrcost$, and $r_i$ remains a \textbf{Type-1/2/4} request. However, the specific type of $r_i$ may change.
\begin{itemize}
\item If $r_i$ is \textbf{Type-4} under correct prediction, it implies that there is no request at other servers between $r_{p(i)}$ and $r_i$, because $l_i = \hyperpar\trfrcost$ is the shortest intended duration of a regular copy. As a result, $r_i$ remains \textbf{Type-4} under misprediction. Thus, the possible type changes of  $r_i$ are from \textbf{Type-1/2} to \textbf{Type-1/2/4} only. By Proposition \ref{costsummary}, the transfer cost allocated to $r_i$ can never increase.
The regular copy cost allocated to $r_i$ increases by $\trfrcost-\hyperpar\trfrcost$. The extension of the regular copy after $r_{p(i)}$ at $s[r_i]$ can only shorten the special copy at $s[r_i]$ or any other server. Hence, the special copy cost allocated to $r_i$ can never increase. Overall, the online cost allocated to $r_i$ can increase by at most $\trfrcost-\hyperpar\trfrcost$.
\item For correct requests, a \textbf{Type-1} (resp. \textbf{Type-3}) correct request remains \textbf{Type-1} (resp. \textbf{Type-3}). Thus, its allocated online cost does not change. A \textbf{Type-2} correct request may change to \textbf{Type-1}, where its allocated online cost can only decrease. A \textbf{Type-4} correct request may change to \textbf{Type-1/2}, where its allocated online cost can increase by at most the cost $\trfrcost$ of a transfer. 
Note that if a request $r_j$ is a \textbf{Type-4} correct request, we must have $t_j - t_{p(j)} > \trfrcost$ (otherwise, the regular copy after $r_{p(j)}$ would be $\trfrcost$ long and $r_j$ must be a \textbf{Type-3} request) and no request arises elsewhere between $r_{p(j)}$ and $r_j$. Thus, at most one \textbf{Type-4} request can change to \textbf{Type-1/2}. 

\item In summary, the total increase in the online cost due to each $r_i \in M_3$ is bounded by $(2 - \hyperpar) \cdot \trfrcost$.
\end{itemize}

By the above analysis, the total increase in the online cost due to all mispredicted requests is bounded by $\trfrcost\cdot |M_2| + (2 - \hyperpar) \cdot \trfrcost \cdot |M_3|$.

Consider an optimal offline strategy of a request sequence $\langle r_{1},r_{2},$ $\dots,r_{m} \rangle$. For each request $r_i$, if $t_i - t_{p(i)} \leq \trfrcost$, by Proposition \ref{prop5}, server $s[r_i]$ holds a copy throughout the period $(t_{p(i)}, t_i)$ to serve $r_i$, which incurs a storage cost of $t_i-t_{p(i)}$.
If $t_i - t_{p(i)} > \trfrcost$, either server $s[r_i]$ holds a copy throughout the period $(t_{p(i)}, t_i)$ to serve $r_i$ or an inward transfer to server $s[r_i]$ is performed to serve $r_i$. In both cases, the cost incurred is at least $\trfrcost$. Moreover, due to the at-least-one-copy requirement, the storage cost incurred between two successive requests $r_{i-1}$ and $r_i$, i.e., over the period $(t_{i-1}, t_i)$, must be at least $t_i - t_{i-1}$. If $t_i - t_{i-1} > \trfrcost$, the portion of the storage cost beyond $\trfrcost$ is not counted in the aforesaid lower bound. Hence, the optimal offline cost is at least $\sum_{i: t_i-t_{p(i)} > \trfrcost} \trfrcost+\sum_{i: t_i-t_{p(i)} \leq \trfrcost}(t_i-t_{p(i)})+\sum_{i:t_i - t_{i-1} > \trfrcost}(t_i - t_{i-1} - \trfrcost)$. Thus, with mispredictions, the increase in the online-to-optimal cost ratio is bounded by
\begin{equation}
\label{eq:panelty}
\frac{\trfrcost\cdot |M_2| + (2 - \hyperpar) \cdot \trfrcost \cdot |M_3|}{\sum\limits_{i: t_i-t_{p(i)} > \trfrcost} \!\!\trfrcost+\!\!\sum\limits_{i: t_i-t_{p(i)} \leq \trfrcost}\!\!(t_i-t_{p(i)})+\!\!\sum\limits_{i:t_i - t_{i-1} > \trfrcost}\!\!(t_i - t_{i-1} - \trfrcost)}.
\end{equation}

In Sections \ref{sec:robustness} and \ref{sec:consistency}, we have shown that Algorithm \ref{alg3} has a tight robustness of $1+\frac{1}{\hyperpar}$ and a tight consistency of $\frac{5+\hyperpar}{3}$. Such a consistency-robustness tradeoff however is not attractive. In order to achieve a better consistency, we have to reduce the hyper-parameter $\hyperpar$, but the robustness goes towards infinity as $\hyperpar$ decreases. On the other hand, when $\hyperpar = 1$ (setting the intended duration of the regular copy after each request to $\trfrcost$, irrespective of the prediction), the consistency and robustness become the same which is $2$. This is in fact the best achievable competitive ratio of a conventional online algorithm \cite{2021Cost}. Note that in our data replication problem, when a request arises, we will find out whether a previous prediction is correct or not. This makes it possible to optimize the robustness-consistency trade-off by changing the $\hyperpar$ setting when needed.
An intuitive idea is to start with setting a small $\hyperpar < 1$ and reset $\hyperpar$ to $1$ once a misprediction is found. Nevertheless, this strategy is not practically useful because it cannot make use of most predictions unless all predictions are correct. Im \textit{et al.} \cite{im2021nonclairvoyant, im2023online} addressed online problems with such temporal aspects by defining prediction error measures based on the difference in the optimal objective value between the predictions and ground truth. In a similar spirit, we propose to leverage the above analysis of mispredictions to monitor (an upper bound of) the online-to-optimal cost ratio as time goes on and adjust $\hyperpar$ dynamically to achieve a good consistency while maintaining a bounded robustness. 

Specifically, we maintain a lower bound of the optimal cost as $\textbf{OPTL} = \sum_{i: t_i-t_{p(i)} > \trfrcost} \trfrcost+\sum_{i: t_i-t_{p(i)} \leq \trfrcost}(t_i-t_{p(i)})+\sum_{i:t_i - t_{i-1} > \trfrcost}(t_i - t_{i-1} - \trfrcost)$, i.e., the denominator of (\ref{eq:panelty}). When a new request $r_i$ arises, we update $\textbf{OPTL}$ based on $t_i - t_{p(i)}$ and $t_i - t_{i-1}$.
We also maintain an upper bound of the online cost $\textbf{OnlineU}$ as the sum of two parts. The first part is the online costs allocated to all requests that have arisen (as given in Proposition \ref{costsummary}). The second part is a conservative estimate of the costs beyond the last seen request at each server by treating the prediction after the last seen request as a misprediction. According to the above analysis, for each $r_i \in M_2$, the correct duration of the regular copy after $r_{p(i)}$ is $\trfrcost$ and the penalty due to misprediction is bounded by $\trfrcost$, so the total cost is at most $2 \trfrcost$; and for each $r_i \in M_3$, the correct duration of the regular copy after $r_{p(i)}$ is $\hyperpar \trfrcost$ and the penalty due to misprediction is bounded by $(2-\hyperpar) \trfrcost$, so the total cost is also at most $2 \trfrcost$. Hence, the second part of the online cost is estimated as $2 \trfrcost n'$, where $n'$ is the number of servers that have received requests. When a new request $r_i$ arises, we update $\textbf{OnlineU}$. Suppose we would like to achieve a robustness of $2 + \beta$ where $\beta \geq 0$. If $\frac{\textbf{OnlineU}}{\textbf{OPTL}} > 2 + \beta$, we apply the conventional online algorithm by setting the intended duration of the regular copy after $r_i$ to $\trfrcost$, irrespective of the prediction. If $\frac{\textbf{OnlineU}}{\textbf{OPTL}} \leq 2 + \beta$, we apply Algorithm \ref{alg3} by setting the intended duration of the regular copy after $r_i$ to $\hyperpar \trfrcost$ or $\trfrcost$ according to the prediction.
In this way, we can maintain a bounded robustness.

\section{Lower Bound of Consistency}
In this section, we establish a lower bound of $\frac{3}{2}$ on the consistency of any deterministic learning-augmented algorithm for our problem. This implies that it is not possible to achieve a consistency approaching $1$ in our problem.

Consider two servers. The inter-request times at each server are always longer than $\trfrcost$. Correct predictions (i.e., forecasting the next request at each server to arise beyond a period of $\trfrcost$ from the previous request) are input to the learning-augmented algorithm all the time. Initially there is only one data copy placed at one server, and a dummy request $r_0$ arises at this server at time 0. The request $r_1$ arrives at the other server right after time 0 (i.e., $t_1=\epsilon$ is a small value), so that in any online or offline strategy, $r_1$ is served by a transfer, and both servers hold data copies immediately thereafter. 

\begin{figure}[t]
\centering
\includegraphics[width=12cm]{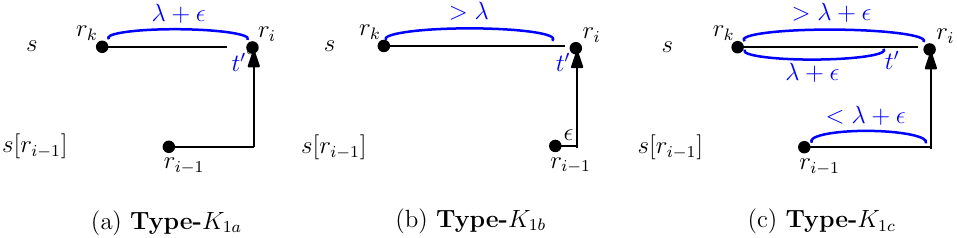}
\caption{\label{typek2} $r_i$ is a \textbf{Type-$K_1$} request}
\end{figure}

The adversary generates subsequent requests according to the behaviour of the online (learning-augmented) algorithm in the following way. After a request $r_{i-1}$ $(i\geq 2)$ (to facilitate presentation, we use $s$ to denote the server other than $s[r_{i-1}]$, and assume the last request at $s$ before $r_{i-1}$ is $r_k$ where $k<i-1$),

\begin{itemize}
\item If $s$ does not hold a data copy at time $t':=\max\{t_{i-1}+\epsilon, t_{k}+\trfrcost+\epsilon\}$ ($\epsilon$ is a small value), the adversary generates the next request $r_i$ in server $s$ at time $t'$, so that $r_i$ must be served by a transfer from $s[r_{i-1}]$. If $t'= t_{k}+\trfrcost+\epsilon$, we call $r_i$ a \textbf{Type-$K_{1a}$} request (see Figure \ref{typek2}(a)).\footnote{There might be other transfers between $s$ and $s[r_{i-1}]$ during the period $(t_{i-1}, t')$ by the online algorithm. To simplify illustration, they are not marked in Figures \ref{typek2} and \ref{typek1}.}
If $t' = t_{i-1}+\epsilon > t_{k}+\trfrcost+\epsilon$, we call $r_i$ a \textbf{Type-$K_{1b}$} request (see Figure \ref{typek2}(b)). 
\item If $s$ holds a data copy at time $t'$, the adversary monitors whether $s$ drops its copy during the period $(t', t_{i-1}+\trfrcost)$. Once $s$ drops its copy at some time instant $t^{*}\in(t', t_{i-1}+\trfrcost)$, the adversary generates the next request $r_i$ in server $s$ at time $t^{*}+\epsilon$ (i.e., right after $s$ drops its copy), so that $r_i$ must be served by a transfer from $s[r_{i-1}]$ (see Figure \ref{typek2}(c)). In this case, we also call $r_i$ a \textbf{Type-$K_{1c}$} request.

In all the above cases, $r_i$ and $r_{i-1}$ are at different servers, so $p(i)=k$. We collectively call $r_i$ a \textbf{Type-$K_1$} request in these cases. It is easy to make the following observation.
\end{itemize}
\begin{observation}
If $r_i$ is a \textbf{Type-$K_{1}$} request, $t_i-t_{i-1}<\trfrcost+\epsilon$. If $r_i$ is a \textbf{Type-$K_{1a}$} request, $t_i-t_{p(i)}=\trfrcost+\epsilon$. If $r_i$ is a \textbf{Type-$K_{1b}$} or \textbf{Type-$K_{1c}$} request, $t_i-t_{p(i)}>\trfrcost+\epsilon$.
\label{requestk2}
\end{observation}
\begin{itemize}
\item If $s$ persistently keeps its copy during the period $(t', t_{i-1}+\trfrcost)$, the adversary generates the next request $r_i$ in server $s[r_{i-1}]$ at time $t_{i-1}+\trfrcost+\epsilon$. In this case, $r_i$ and $r_{i-1}$ are at the same server. We call such $r_i$ a \textbf{Type-$K_{2}$} request. If $s[r_{i-1}]$ holds a data copy when $r_i$ arises, $r_i$ is served locally (see Figure \ref{typek1}(a)). Otherwise, $r_i$ is served by a transfer from $s$ (see Figure \ref{typek1}(b)). 
\end{itemize}

\begin{observation}
If $r_i$ is a \textbf{Type-$K_{2}$} request, $t_i - t_{i-1}=\trfrcost+\epsilon$.
\label{requestk1}
\end{observation}

\begin{figure}[htbp]
\centering
\includegraphics[width=9cm]{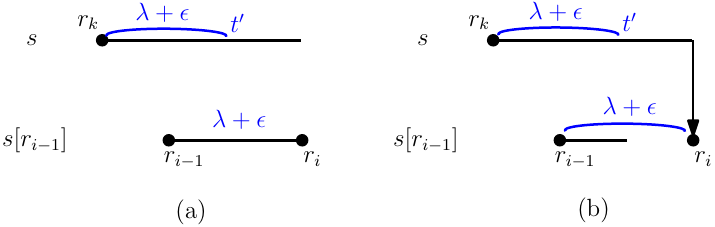}
\caption{\label{typek1} $r_i$ is a \textbf{Type-$K_2$} request}
\end{figure}

We use $\textbf{online}(r_i, r_{j})$ to denote the total storage and transfer cost produced by the online algorithm over the period $(t_{i}, t_j]$. The proof of the following observation is given in Appendix \ref{prooflb}.

\begin{observation}\
\begin{itemize}
\item If $r_i$ is a \textbf{Type-$K_{1a}$} request, $\textbf{online}(r_{i-1}, r_{i})\geq(t_i-t_{i-1})+\trfrcost$.
\item If $r_i$ is a \textbf{Type-$K_{1b}$} or \textbf{Type-$K_{1c}$} request, $\textbf{online}(r_{i-1}, r_{i})\geq2(t_i-t_{i-1})+\trfrcost-\epsilon$.
\item If $r_i$ is a \textbf{Type-$K_2$} request, $\textbf{online}(r_{i-1}, r_{i})\geq2\trfrcost+\epsilon$.
\end{itemize}
\label{interobs}
\end{observation}

To analyze the online-to-optimal cost ratio, we partition the whole request sequence into subsuquences. Each subsequence consists of one or two requests.

\textbf{Case 1:} A subsequence $\langle r_{i-1}, r_i \rangle$, where $r_{i-1}$ is \textbf{Type-$K_{2}$} and $r_{i}$ is \textbf{Type-$K_{1}$}.

\textbf{Case 2:} A subsequence $\langle r_{i-1}, r_i \rangle$ of two successive \textbf{Type-$K_{1}$} requests. 

\textbf{Case 3:} A subsequence $\langle r_i \rangle$ of a single \textbf{Type-$K_{2}$} request.

For each subsequence, we construct an offline strategy and show that the cost ratio between the online algorithm and the offline strategy is at least $\frac{3}{2}$ or can be made arbitrarily close to $\frac{3}{2}$ as $\epsilon \rightarrow 0$.
This then gives rise to a lower bound $\frac{3}{2}$ on the online-to-optimal cost ratio for the whole request sequence.
Please refer to Appendix \ref{subsequence} for details.
Therefore, the consistency of any deterministic learning-augmented algorithm is at least $\frac{3}{2}$.

\section{Experimental Evaluation}

We conduct simulation experiments to evaluate our algorithms using real data access traces. The results (1) verify our robustness and consistency analysis; (2) show that our algorithms can make effective use of predictions to improve performance; and (3) demonstrate that the adapted Algorithm \ref{alg3} (presented in Section \ref{misprediction}) can achieve a specified robustness. Due to space limitations, we defer the detailed experimental setup and results to Appendix \ref{experiment}.

\section{Concluding Remarks}
\label{remarks}

Note that when $\hyperpar = 1$ (i.e., not using predictions), our algorithm achieves a competitive ratio of $2$, which is better than the competitive ratio of $3$ by the algorithm proposed in Wang \textit{et al.}~\cite{wang2018cost}. 

Recently, Wang \textit{et al.} \cite{2021Cost} presented an online algorithm for a multi-server system where the servers can have distinct storage costs, and claimed that the algorithm is $2$-competitive. We find that the claim is not true. The competitive ratio of this algorithm is at least $\frac{5}{2}$, even if all servers have the same storage cost.

Let $\mu(s_i)$ denote the cost for storing a data copy in server $s_i$ per time unit. Assume that the servers are indexed in ascending order of storage cost rate, i.e., $\mu(s_1) \leq \mu(s_2) \leq \cdots \leq \mu(s_n)$. The main idea of Wang \textit{et al.}'s algorithm is as follows.
\begin{itemize}
\item For each server $s_{i}$, after serving a local request (either by a transfer or by a local copy), $s_{i}$ keeps the data copy for $\frac{\trfrcost}{\mu(s_i)}$ time units. Note that the cost of storing the copy in $s_i$ over this period matches the cost of transferring the object.
\item If a new request arises at $s_i$ in this period, the request is served by the local copy and $s_i$ keeps the copy for another $\frac{\trfrcost}{\mu(s_i)}$ time units (starting from the new request).
\item When the data copy in $s_{1}$ expires (i.e., the server with the lowest storage cost rate), the algorithm checks whether $s_{1}$ holds the only copy in the system. If so, $s_{1}$ continues to keep the copy for another $\frac{\trfrcost}{\mu(s_1)}$ time units. Otherwise, $s_{1}$ drops its copy.
\item When the data copy in $s_{i}$ $(i \neq 1)$ expires, the algorithm also checks whether $s_{i}$ holds the only copy. If not, $s_{i}$ drops its copy. If $s_i$ holds the only copy, the algorithm further checks whether $s_{i}$ has kept the copy for $\frac{\trfrcost}{\mu(s_i)}$ time units since its most recent local request (i.e., there is no request at $s_i$ for $\frac{\trfrcost}{\mu(s_i)}$ time). If so, $s_i$ continues to keep the copy for another $\frac{\trfrcost}{\mu(s_i)}$ time units. Otherwise, it implies that $s_{i}$ has kept the copy for $\frac{2\trfrcost}{\mu(s_i)}$ time units without serving any local request. In that case, $s_{i}$ transfers the object to $s_{1}$ and drops the local copy.
\end{itemize}

Figure \ref{wang} gives a counterexample to the claim of the competitive ratio $2$, where two servers $s_{1}$ and $s_{2}$ both have storage cost rates $1$.
There are $m$ requests. Requests $r_1$ arises at $s_{1}$ and all subsequent requests arise at $s_{2}$.
The arising times of the requests are $t_{1}=0$, $t_{2}=\epsilon$, $t_{3}=2\trfrcost+2\epsilon$, $t_{4}=4\trfrcost+3\epsilon$, $t_{5}=6\trfrcost+4\epsilon$ and so on,\footnote{To simplify boundary conditions, Wang \textit{et al.} \cite{2021Cost} assumed that (1) the object is initially stored in $s_1$ (the server with the lowest storage cost rate)
and (2) the first request arises at $s_1$ at time $0$. Our example follows these assumptions.} where $\epsilon > 0$ is a small value and every two consecutive requests at $s_2$ are slightly longer than $2\trfrcost$ apart.

\begin{figure}[t]
\centering
\includegraphics[width=8cm]{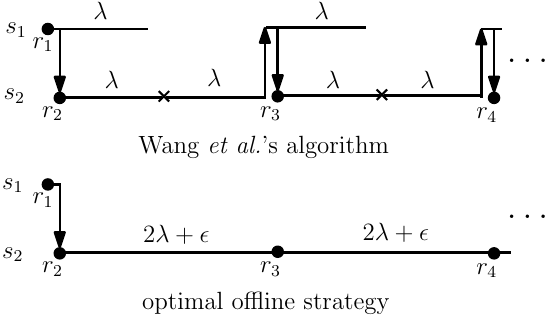}
\caption{\label{wang} A counterexample of Wang \textit{et al.}'s algorithm}
\end{figure}

By Wang \textit{et al.}'s algorithm, after request $r_{1}$, the copy in $s_{1}$ would expire at time $t_1 + \trfrcost = \trfrcost$, and after request $r_2$, the copy in $s_2$ would expire at time $t_2 + \trfrcost = \trfrcost + \epsilon > \trfrcost$ which is after $s_1$'s copy expires. Hence, $s_1$ drops its copy when it expires. When $s_2$'s copy expires, since it is the only copy in the system, it is renewed for another $\trfrcost$ time units.
When the renewal expires, $s_2$ transfers the object to $s_1$ and drops the local copy. Then, the copy in $s_1$ would expire at time $3\trfrcost+\epsilon$. When request $r_{3}$ arises at $s_2$, $s_1$ transfers the object to $s_2$ to serve $r_3$, after which the copy in $s_2$ would expire at time $t_3 + \trfrcost = 3\trfrcost + 2\epsilon > 3\trfrcost + \epsilon$ which is again after $s_1$'s copy expires. Hence, $s_1$ drops its copy when it expires. When $s_2$'s copy expires, it is renewed for another $\trfrcost$ time units. When the renewal expires, $s_2$ transfers the object to $s_1$ and drops the local copy. Then, the copy in $s_1$ would expire at time $5\trfrcost+2\epsilon$.
When request $r_{4}$ arises at $s_2$, $s_1$ transfers the object to $s_2$ to serve $r_4$, after which the copy in $s_2$ would expire at time $t_4 + \trfrcost = 5\trfrcost + 3\epsilon > 5\trfrcost + 2\epsilon$ which is again after $s_1$'s copy expires.
This pattern is repeated continuously. As a result, the total cost of serving all requests is at least $(m-2)\cdot5\trfrcost$ (where we include the cost incurred up to the final request $r_m$ only for fair comparison with the optimal offline solution).

In the optimal offline solution, server $s_2$ should always keep a data copy after request $r_2$ is served by a transfer. So, the total cost is $(m-2)\cdot (2\trfrcost+\epsilon) + \trfrcost + \epsilon$.

It is easy to see that as $m\rightarrow\infty$ and $\epsilon\rightarrow 0$, the ratio between the online cost and the optimal offline cost approaches $\frac{5}{2}$.

\newpage
\appendix
\section{Proof of Proposition \ref{pro-1}}
\label{proofreqatsrcdes}
\begin{proof}
Since one data copy is initially placed in server $s_1$
where the dummy request $r_0$ arises at time $0$, no transfer is needed to serve $r_0$.
Thus, we consider only transfers done after time $0$.

Suppose that a transfer from a server $s_x$ to a server $s_y$ is
carried out at time $t$ in an optimal replication strategy, and
there is no request at servers $s_x$ and $s_y$ at time $t$.
It can be inferred that either (1) a data copy is held by $s_x$ before
the transfer, or (2) $s_x$ just receives the data object from
another server $s_z$ at time $t$. In the latter case, since there is
no request at $s_x$ at time $t$, the transfer from $s_x$ to $s_y$
can be replaced by a transfer from $s_z$ to $s_y$ without affecting
the cost of the replication strategy (see Figure \ref{trans1} for an illustration). If $s_z$ does not have a local
request at time $t$ and does not hold a copy before time $t$, we
can find another server transferring to $s_z$ at time $t$ and repeat
the replacement until a source server of the transfer having a local
request at time $t$ or holding a copy before time $t$ is found.
So without loss of generality, we  assume that $s_x$ holds a copy before the transfer.

\begin{figure}[htbp]
\centering
\includegraphics[width=6.5cm]{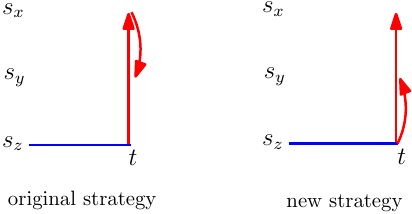}
\caption{\label{trans1} Source server holds a copy before the transfer}
\end{figure}

Similarly, since there is no request at $s_y$ at time $t$, if $s_y$
does not hold a copy after the transfer, $s_y$ must be sending
the data object to another server $s_z$ at time $t$ (otherwise, the
transfer from $s_x$ to $s_y$ can be removed, which contradicts the
optimality of the replication strategy). We can then replace the
transfers from $s_x$ to $s_y$ and from $s_y$ to $s_z$ by a single
transfer from $s_x$ to $s_z$ without affecting the ability of the
replication strategy to serve all the requests, which again
contradicts the optimality of the replication strategy (see Figure \ref{trans2} for an illustration). Thus, $s_y$ must hold a copy after the transfer.

\begin{figure}[htbp]
\centering
\includegraphics[width=7cm]{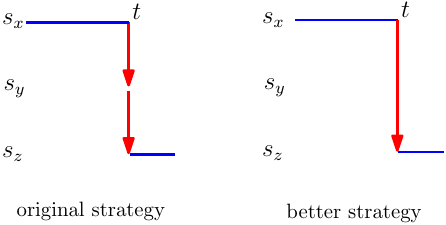}
\caption{\label{trans2} Destination server holds a copy after the transfer}
\end{figure}

If $s_x$ continues to hold a copy after the transfer, we can
delay the transfer from $s_x$ to $s_y$ by some small $\epsilon >
0$ and create a copy at $s_y$ at time $t+\epsilon$ to save
the storage cost (see Figure \ref{delaytransfer} for an illustration). This would not affect the feasibility of the
replication strategy since all the transfers originating from $s_y$
during the period $(t, t+\epsilon)$ can originate from $s_x$
instead. As a result, it contradicts the optimality of the replication strategy.
\begin{figure}[htbp]
\centering
\includegraphics[width=7cm]{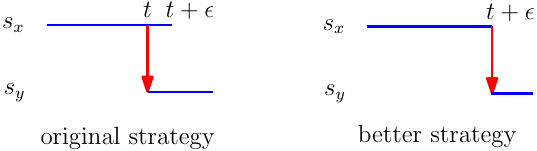}
\caption{\label{delaytransfer} Transfer can be delayed if $s_x$ keeps its copy after $t$}
\end{figure}

Now suppose that $s_x$ does not hold a copy after the transfer. Then it can be shown that the copy held by $s_x$ before the transfer serves at least one local request at $s_x$. Otherwise, if the copy held by $s_x$ before the transfer
does not serve any local request at $s_x$, let $t^\prime$ denote the creation
time of the copy. At time $t^\prime$, there must be a transfer
from another server $s_z$ to $s_x$. We can replace the two transfers from $s_z$ to $s_x$ at time $t^\prime$ and
from $s_x$ to $s_y$ at time $t$ by one transfer from $s_z$ to $s_y$ at time
$t^\prime$, remove the copy at $s_x$ and add a copy at $s_y$
during $(t^\prime, t)$, which reduces the total cost,  contradicting the optimality of the replication strategy (see Figure \ref{eg1} for an illustration). Thus, the copy held by $s_x$ before the transfer serves at least one local request at $s_x$.
We look for the last such local request of $s_x$. Let $t^*$ denote the time of this request. Then we can
bring the transfer from $s_x$ to $s_y$ earlier to time $t^*$,
remove the copy at $s_x$ and add a copy at $s_y$ during the
period $(t^*, t)$ (see Figure \ref{eg2} for an illustration). This would not increase the cost of the
replication strategy, and would not affect the service of other requests because all transfers originating from $s_x$ during $(t^*, t)$ can originate from $s_y$ instead.
In the new strategy, there is a request at $s_x$ at the time of the transfer from $s_x$ to $s_y$.
\end{proof}

\begin{figure}[htbp]
\centering
\includegraphics[width=7cm]{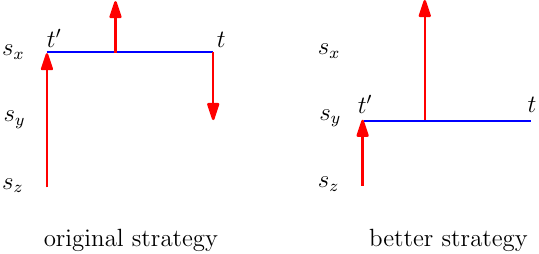}
\caption{\label{eg1} The copy at $s_x$ must serve at least one local request}
\end{figure}

\begin{figure}[htbp]
\centering
\includegraphics[width=7cm]{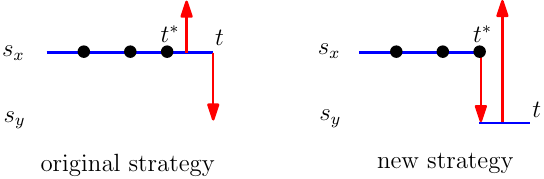}
\caption{\label{eg2} Change the strategy to have a request at the source server of the transfer}
\end{figure}

\section{Proof of Proposition \ref{pro-2}}
\begin{proof}
In an optimal replication strategy satisfying Proposition \ref{pro-1}, suppose $r_i$ is served by a local copy, but the copy is created later than $r_{p(i)}$. This implies that the copy must be created by a transfer (see Figure \ref{a} for an illustration). By Proposition \ref{pro-1}, there must be a request $r_j$ at the source server of this transfer. Then, we can replace the copy at $s[r_i]$ during $(t_j, t_i)$ with a copy at $s[r_j]$ during $(t_j, t_i)$, and delay the transfer to the time $t_i$ of $r_i$. This would not affect the service of other requests, because all transfers originating from $s[r_i]$ during
this period can originate from $s[r_j]$ instead. As a result, the total cost does not change.
In the new strategy, $r_i$ is served by a transfer, and
the characteristic in Proposition \ref{pro-1} is retained.
\end{proof}

\begin{figure}[htbp]
\centering
\includegraphics[width=9.5cm]{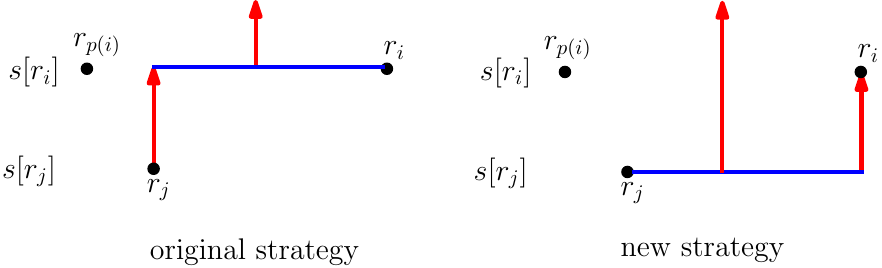}
\caption{\label{a} Illustration for the proof of Proposition \ref{pro-2}}
\end{figure}

\section{Proof of Proposition \ref{prop5}}
\begin{proof}
In an optimal offline strategy satisfying Propositions \ref{pro-1} and \ref{pro-2}, if server $s[r_i]$ does not hold a copy throughout the period $(t_{p(i)}, t_i)$,
$s[r_i]$ must receive a transfer during $(t_{p(i)}, t_i)$
in order to serve request $r_i$,
where the transfer cost incurred is $\trfrcost$.
Since $t_i - t_{p(i)} \leq \trfrcost$, we can replace the transfer with a copy at $s[r_i]$ during $(t_{p(i)}, t_i)$ without increasing the total cost. Such replacement retains the characteristics in Propositions \ref{pro-1} and \ref{pro-2}.
\end{proof}

\section{Proof of Proposition \ref{prop-add}}
\label{prooftransfer}
\begin{proof}
In an optimal replication strategy satisfying Propositions \ref{pro-1}, \ref{pro-2} and \ref{prop5}, if the source server of the transfer is $s[r_{i-1}]$, (i) holds naturally. By assumption, since no server holds a copy crossing $t_i$, $s[r_{i-1}]$ must drop its copy after the transfer. If $s[r_{i-1}]$ does not keep its copy since $t_{i-1}$, the copy must be created by a transfer after $t_{i-1}$ (see Figure \ref{new3} for an illustration). By Proposition \ref{pro-1}, there must be a request $r_j$ at the source server of this transfer. This leads to a contradiction, because $r_{i-1}$ and $r_i$ are two consecutive requests in the sequence. Hence, $s[r_{i-1}]$ must keep its copy since $t_{i-1}$, so (ii) also holds.

\begin{figure}[htbp]
\centering
\includegraphics[width=4cm]{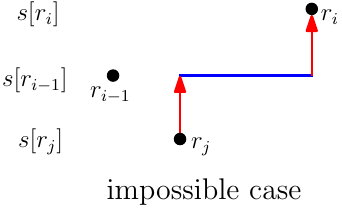}
\caption{\label{new3} Illustration that $s[r_{i-1}]$ must keep a copy since $t_{i-1}$}
\end{figure}

Now suppose that the source server of the transfer is a server $s$ other than $s[r_{i-1}]$. By assumption, since no server holds a copy crossing $t_i$, $s$ must drop its copy after the transfer. It can be proved that the copy in $s$ must be created no later than the last request $r_h$ at $s$ before $r_i$. Otherwise, the copy in $s$ does not serve any local request at $s$, so the copy in $s$ must be created by a transfer (see Figure \ref{new0} for an illustration). By Proposition \ref{pro-1}, there must be a request $r_j$ at the source server of this transfer. Then, we can replace the copy at $s$ during $(t_j, t_i)$ with a copy at $s[r_j]$ during $(t_j, t_i)$, replace the transfer from $s$ to $s[r_i]$ with a transfer from $s[r_j]$ to $s[r_i]$, and remove the transfer from $s[r_j]$ to $s$. This would not affect the service of other requests, because all transfers originating from $s$ during this period can originate from $s[r_j]$ instead. As a result, the total cost is reduced, contradicting the optimality of the strategy. Thus, the copy in $s$ must be created no later than the last request $r_h$ at $s$ before $r_i$. Note that since $s \neq s[r_{i-1}]$, we have $h<i-1$.

\begin{figure}[htbp]
\centering
\includegraphics[width=8.5cm]{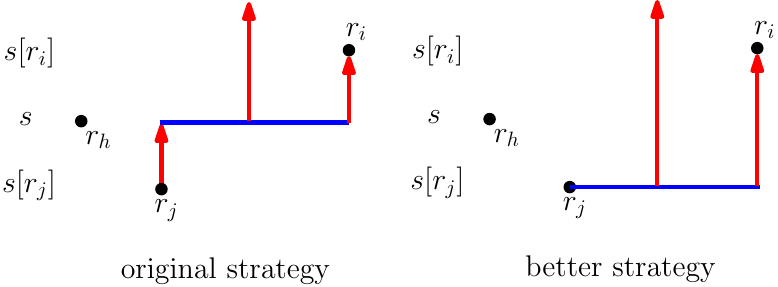}
\caption{\label{new0} The copy in $s$ must serve a local request}
\end{figure}

If $r_{i-1}$ and $r_i$ arise at the same server, we can replace the copy at $s$ during $(t_{i-1}, t_i)$ with a copy at $s[r_i]$ during $(t_{i-1}, t_i)$ (see Figure \ref{new1} for an illustration). This replacement can save a transfer cost of $\trfrcost$, which contradicts the optimality of the strategy. Hence, $r_{i-1}$ and $r_i$ must arise at different servers. Then, we can replace the copy at $s$ during $(t_{i-1}, t_i)$ with a copy at $s[r_{i-1}]$ during $(t_{i-1}, t_i)$, and change the transfer for serving $r_i$ to originate from $s[r_{i-1}]$ (see Figure \ref{new2} for an illustration). This does not affect the total cost of the strategy, and both (i) and (ii) hold in the new strategy. Meanwhile, the characteristics in Propositions \ref{pro-1}, \ref{pro-2} and \ref{prop5} are retained in the new strategy.
\end{proof}

\begin{figure}[htbp]
\centering
\includegraphics[width=7.5cm]{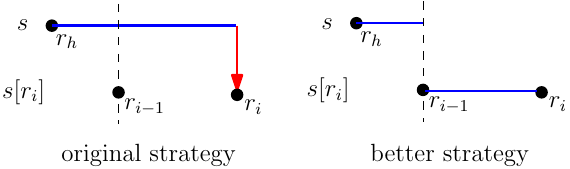}
\caption{\label{new1} $r_{i-1}$ and $r_i$ arise at the same server}
\end{figure}
\begin{figure}[htbp]
\centering
\includegraphics[width=8cm]{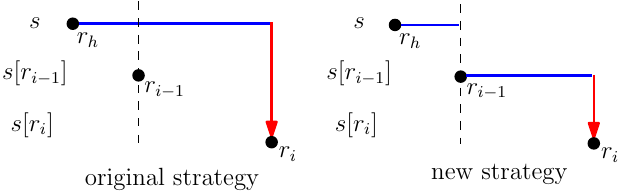}
\caption{\label{new2} $r_{i-1}$ and $r_i$ arise at different servers}
\end{figure}

\section{Proof of Lemma \ref{lem-2}}
\label{proofspecial}
\begin{proof}
By our online algorithm, each special copy occurs only after the expiration of a regular copy which has a duration of $\hyperpar\trfrcost$ according to Proposition \ref{consistrequest}. It follows from Proposition \ref{pro-3} that all special copies must be at least $\hyperpar\trfrcost$ apart from each other.
In addition, each period $(t_{p(i)}, t_{i})\in X$ covers all special copies therein and the regular copy after $r_{p(i)}$ (which has a duration $l_{i} = \hyperpar\trfrcost$ by Propositions \ref{prop:rx} and \ref{consistrequest}). Thus, if a period $(t_{p(i)}, t_{i})\in X$ covers $y$ special copies, the total duration of these special copies is bounded by $(t_{i} - t_{p(i)})-y\cdot\hyperpar\trfrcost$. Since each special copy is used to serve one request in $R_2\cup R_4$ so that there are $|R_2|+|R_4|$ special copies in total, the lemma follows by aggregating the bounds for all the periods $(t_{p(i)}, t_{i})\in X$:
\begin{eqnarray*}
    \sum_{r_i\in R_2\cup R_4}\left(t_{i}-t'_{i}\right) & \leq & \sum_{(t_{p(i)},t_i)\in X}(t_{i} - t_{p(i)})-\left(|R_2|+|R_4|\right)\cdot\hyperpar\trfrcost \\
    & = & \sum_{r_i\in R_X}(t_{i} - t_{p(i)})-\left(|R_2|+|R_4|\right)\cdot\hyperpar\trfrcost.
\end{eqnarray*}
\end{proof}

\section{Proof of Lemma \ref{lem-5}}
\begin{proof}
In the case of $|X|\geq 2$, by Proposition \ref{prop:rx}, the end sentinel requests $R_X$ of all storage periods in $X$ are \textbf{Type-1/2/4} requests.

In the optimal offline strategy, following the discussion in Section \ref{sec:preparation}, for each storage period $(t_{p({k_j})}, t_{k_j})\in X$ where $j < |Q|$, $(t_{p({k_j})}, t_{k_j})$  crosses the start time  $t_{p({k_{j+1}})}$ of the next period $(t_{p({k_{j+1}})},$ $t_{k_{j+1}})\in Q$. In server $s[r_{k_{j+1}}]$, there must be a \textbf{Type-1/2/4} request in the period $(t_{p({k_j})}, t_{p(k_{j+1})}]$. This is because if all requests in $s[r_{k_{j+1}}]$ during $(t_{p({k_j})}, t_{p(k_{j+1})}]$ are \textbf{Type-3} requests, by definition, a regular copy is stored in $s[r_{k_{j+1}}]$ throughout $(t_{p({k_j})}, t_{p(k_{j+1})}]$ in our online algorithm. By Proposition \ref{pro-3}, it implies that there is no special copy in the period $(t_{p({k_j})},$ $t_{p(k_{j+1})}]$, which contradicts the assumption $(t_{p({k_j})}, t_{k_j})\in X$.

Let $r_g$ denote the last \textbf{Type-1/2/4} request in server $s[r_{k_{j+1}}]$ in the period $(t_{p({k_j})}, t_{p(k_{j+1})}]$. Then, there is a copy in server $s[r_{k_{j+1}}]$ throughout $(t_g, t_{p(k_{j+1})}]$ (since all requests therein are \textbf{Type-3}) and hence $(t_g, t_{k_{j+1}})$. We can show that $r_g \notin R_X$ by contradiction. Assume on the contrary that $r_g \in R_X$, i.e., $(t_{p(g)},t_{g})\in X$.
Then, there is a copy in server $s[r_{k_{j+1}}]$ throughout $(t_{p(g)}, t_{k_{j+1}})$ in the optimal offline strategy (see Figure \ref{RR} for an illustration). In addition, there is also a copy in server $s[r_{k_j}]$ during $(t_{p(k_j)},t_{k_j})$.
By Propositions \ref{prop:rx} and \ref{consistrequest}, both periods $(t_{p(k_j)},t_{k_j})$ and $(t_{p(g)},t_{g})$ are longer than $\trfrcost$.
If $t_{p(g)}<t_{p(k_j)}$,  the storage period $(t_{p(k_j)},t_{k_j})$ can be replaced by a transfer from $s[r_{k_{j+1}}]$ to $s[r_{k_j}]$ to serve $r_{k_j}$ which saves cost.
If $t_{p(g)}>t_{p(k_j)}$, the storage period $(t_{p(g)},t_{g})$ can be replaced by a transfer from $s[r_{k_j}]$ to $s[r_{k_{j+1}}]$ to serve $r_{g}$ which saves cost.
In either case, it contradicts the optimality of the offline strategy. Therefore, we must have $r_{g}\notin R_X$. Moreover, note that the request $r_g$ found falls in the period $(t_{p({k_j})}, t_{p(k_{j+1})}]$. Thus, all the requests $r_g$'s for different storage periods in $X$ are distinct.

\begin{figure}[htbp]
\centering
\includegraphics[width=12cm]{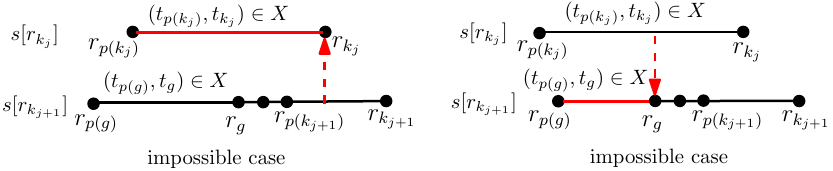}
\caption{Illustration for the case of $|X| \geq 2$}
\label{RR}
\end{figure}

In summary, for each storage period $(t_{p(k_j)},t_{k_j})\in X$ where $j<|Q|$, we can find two \textbf{Type-1/2/4} requests $r_g$ and $r_{k_j}$. If the storage period $(t_{p(k_{|Q|})},t_{k_{|Q|}})\in X$, we can find one \textbf{Type-1/2/4} request $r_{k_{|Q|}}$. All the above requests are distinct. Thus, in total, there are at least $(2|X|-1)$ \textbf{Type-1/2/4} requests.

In the case of $|X|=1$, suppose $(t_{p(k_j)},t_{k_j})$ is the only storage period in $X$. If $j<|Q|$, similar to the above arguments, we can find two \textbf{Type-1/2/4} requests. If $j=|Q|$, by Proposition \ref{prop:rx}, the end sentinel request $r_{k_{|Q|}}\in R_X$ is a \textbf{Type-1/2/4} request. To find another \textbf{Type-1/2/4} request, note that we have two different servers $s[r_{k_{|Q|}}]$ and $s[r_{k_{|Q|-1}}]$ since $|Q| \geq 2$. At least one of them is not the server $s[r_d]$ holding the copy over the first storage period $(t_{p(k_1)},t_{k_1})$ in $Q$ (where $p(k_1) = d$) in the optimal offline strategy. Let $s_g$ denote that server. Let $r_g$ denote the first request at $s_g$ after time $t_{p(k_1)} = t_d$. By the partition definition, there is no copy at $s_g$ crossing time $t_d$, so $r_g$ must be served by a transfer in the optimal offline strategy. By Proposition \ref{prop5}, $t_g-t_{p(g)}>\trfrcost$. It follows from Proposition \ref{consistrequest} that $r_g$ is a \textbf{Type-1/2/4} request by our online algorithm. Together with $r_{k_{|Q|}}$, we have two \textbf{Type-1/2/4} requests.
\end{proof}

\section{Proof of Lemma \ref{Lem-b}}
\begin{proof}
First, we prove the existence of at least one request during $(t_{p(k_1)},t_{k_1})$ at some server other than $s[r_{k_1}]$. Since $r_{k_1}$ is a \textbf{Type-1/2} request, by Proposition \ref{consistrequest}, $t_{k_1}-t_{p(k_1)}>\trfrcost$ and the duration of the regular copy after $r_{p(k_1)}$ is $\hyperpar\trfrcost$. Suppose there is no request during $(t_{p(k_1)},t_{k_1})$ at any server other than $s[r_{k_1}]$. For each server $s\neq s[r_{k_1}]$, we use $r_{h}$ to denote its first local request after $t_{k_1}$, then $r_{p(h)}$ is its last local request before $t_{p(k_1)}$, i.e., $t_{p(h)}<t_{p(k_1)}<t_{k_1}<t_h$ (see Figure \ref{AAAA5} for an illustration). Since $t_{k_1}-t_{p(k_1)}>\trfrcost$, we  have $t_h-t_{p(h)}>\trfrcost$ and hence by Proposition \ref{consistrequest}, the duration of the regular copy after $r_{p(h)}$ is $\hyperpar\trfrcost$. This regular copy must expire before the regular copy after $r_{p(k_1)}$ since $t_{p(h)}<t_{p(k_1)}$. Thus, $s[r_{k_1}]$ must hold the only copy when its regular copy expires, so its regular copy will switch to a special copy and $r_{k_1}$ will be served by a local copy. This implies that $r_{k_1}$ is a \textbf{Type-4} request, which contradicts the assumption that $r_{k_1}$ is a \textbf{Type-1/2} request.

\begin{figure}[htbp]
\centering
\includegraphics[width=4.5cm]{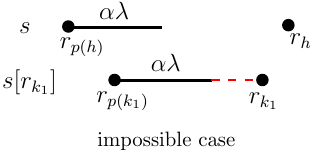}
\caption{Illustration for the proof of Lemma \ref{Lem-b}}
\label{AAAA5}
\end{figure}

Next, consider a server $s\neq s[r_{k_1}]$ with requests arising during $(t_{p(k_1)},t_{k_1})$.
Let $r_g$ denote the first request at $s_g$ after time $t_{p(k_1)} = t_d$
(the beginning time of the partition).
By the partition definition, there is no copy at $s$ crossing time $t_d$, so $r_g$ must be served by a transfer in the optimal offline strategy. By Proposition \ref{prop5}, $t_g-t_{p(g)}>\trfrcost$. It follows from Proposition \ref{consistrequest} that $r_g$ is a \textbf{Type-1/2/4} request by our online algorithm. Together with $r_{k_1}$, we have at least two \textbf{Type-1/2/4} requests.
\end{proof}

\section{Proof of Observation \ref{interobs}}
\label{prooflb}
\begin{proof}
If $r_i$ is a \textbf{Type-$K_{1a}$} request, $\textbf{online}(r_{i-1}, r_{i})$ includes a transfer cost of $\trfrcost$ (for serving $r_i$) and a storage cost at least $t_i - t_{i-1}$ (since there is at least one data copy at any time).

If $r_i$ is a \textbf{Type-$K_{1b}$} request, by the same arguments, $\textbf{online}(r_{i-1}, r_{i})$ is at least $(t_i - t_{i-1}) + \trfrcost = \epsilon + \trfrcost = 2(t_i - t_{i-1}) + \trfrcost - \epsilon$.

If $r_i$ is a \textbf{Type-$K_{1c}$} request, $\textbf{online}(r_{i-1}, r_{i})$ can be divided into three parts: (1) the storage cost at $s$ and the cost of transfers from $s[r_{i-1}]$ to $s$ (if any) during $(t_{i-1}, t_i)$, (2) the storage cost at $s[r_{i-1}]$ and the cost of transfers from $s$ to $s[r_{i-1}]$ (if any) during $(t_{i-1}, t_i)$, (3) the cost $\trfrcost$ of the transfer to serve $r_i$ at time $t_i$. For part (1), by definition, $s$ keeps a copy during $(t', t_{i}-\epsilon)$, where $t'=\max\{t_{i-1}+\epsilon, t_{p(i)}+\trfrcost+\epsilon\}$. If $s$ persistently keeps its copy throughout the period $(t_{i-1}, t_i-\epsilon)$, its storage cost is $t_i-t_{i-1}-\epsilon$. Otherwise, there must be a transfer from $s[r_{i-1}]$ to $s$ for $s$ to hold a copy from $t'$ onward, which incurs a transfer cost $\trfrcost > t_i-t_{i-1}-\epsilon$ by Observation \ref{requestk2}. Hence, the cost of part (1) is at least $t_i-t_{i-1}-\epsilon$. For part (2), since $r_i$ is served by a transfer from $s[r_{i-1}]$, it implies that $s[r_{i-1}]$ holds a copy at time $t_i$. If $s[r_{i-1}]$ persistently keeps its copy throughout the period $(t_{i-1}, t_i)$, its storage cost is $t_i-t_{i-1}$. Otherwise, there must be a transfer from $s$ to $s[r_{i-1}]$ for $s[r_{i-1}]$ to hold a copy at $t_i$, which incurs a transfer cost $\trfrcost>t_i-t_{i-1}-\epsilon$ by Observation \ref{requestk2}. Moreover, since $s$ does not hold a copy during $(t_i-\epsilon, t_i)$, $s[r_{i-1}]$ must hold a copy during this period, which implies that the storage cost at $s[r_{i-1}]$ is at least $\epsilon$. Hence, the cost of part (2) is at least $t_i-t_{i-1}$. Adding up parts (1), (2) and (3), we have $\textbf{online}(r_{i-1}, r_{i})\geq2(t_i-t_{i-1})+\trfrcost-\epsilon$.

In the case that $r_i$ is a \textbf{Type-$K_{2}$} request, if there is at least one transfer between $s[r_{i-1}]$ and $s$ during $(t_{i-1}, t_{i})$, $\textbf{online}(r_{i-1}, r_{i})$ includes a transfer cost of $\trfrcost$ and a storage cost at least $t_i - t_{i-1}$ (since there is at least one data copy at any time). Otherwise, if there is no transfer, $s[r_{i-1}]$ must persistently keep its copy during $(t_{i-1}, t_{i})$ to serve $r_i$ locally, which incurs a storage cost $t_i - t_{i-1}$. Meanwhile, by definition, $s$ persistently keeps its copy during $(t', t_{i-1}+\trfrcost+\epsilon)$. Since there is no transfer, then it must persistently keep its copy during $(t_{i-1}, t_{i-1}+\trfrcost+\epsilon)$, which incurs a storage cost $\trfrcost+\epsilon$. Hence, $\textbf{online}(r_{i-1}, r_{i})$ includes a total storage cost at least $t_i - t_{i-1} + \trfrcost = 2\trfrcost+\epsilon$ (Observation \ref{requestk1}).
\end{proof}

\section{Analysis of Subsequences}
\label{subsequence}

Similar to $\textbf{online}(r_{i}, r_{j})$, we use $\textbf{offline}(r_{i}, r_{j})$ to denote the total storage and transfer cost incurred over the period $(t_{i}, t_j]$ in the offline strategy constructed.

\textbf{Case 1:} A subsequence $\langle r_{i-1}, r_i \rangle$, where $r_{i-1}$ is \textbf{Type-$K_{2}$} and $r_{i}$ is \textbf{Type-$K_{1}$}. By definition, $r_{i-2}$ and $r_{i-1}$ are at the same server. By Observation \ref{requestk1}, $t_{i-1}-t_{i-2}=\trfrcost+\epsilon$. In an offline strategy, $s[r_{i-1}]$ can keep its copy throughout the period $(t_{i-2}, t_i)$ and serve $r_{i}$ by a transfer (see Figure \ref{subseq}(a)). Hence, $\textbf{offline}(r_{i-2}, r_i)\leq (t_i-t_{i-2})+\trfrcost=(t_i-t_{i-1})+2\trfrcost+\epsilon$.

Since $r_{i-1}$ is a \textbf{Type-$K_{2}$} request, $\textbf{online}(r_{i-2}, r_{i-1})\geq 2\trfrcost+\epsilon$ by Observation \ref{interobs}. Since $r_i$ is at a different server from $r_{i-2}$ and $r_{i-1}$, we must have $p(i)<i-2$. It follows from $t_{i-1}-t_{i-2}=\trfrcost+\epsilon$ that $t_{i}-t_{p(i)}>t_{i-1}-t_{i-2}=\trfrcost+\epsilon$. By Observation \ref{requestk2}, $r_i$ must be a \textbf{Type-$K_{1b}$} or \textbf{Type-$K_{1c}$} request. Then by Observation \ref{interobs}, $\textbf{online}(r_{i-1}, r_{i})\geq2(t_i-t_{i-1})+\trfrcost-\epsilon$. Hence, $\textbf{online}(r_{i-2}, r_{i})=\textbf{online}(r_{i-2}, r_{i-1})+\textbf{online}(r_{i-1}, r_{i})\geq2(t_i-t_{i-1})+3\trfrcost$. Thus, $\frac{\textbf{online}(r_{i-2}, r_{i})}{\textbf{offline}(r_{i-2}, r_i)}\geq\frac{2(t_i-t_{i-1})+3\trfrcost}{(t_i-t_{i-1})+2\trfrcost+\epsilon}$, which approaches at least $\frac{3}{2}$ as $\epsilon \rightarrow 0$.

\begin{figure}[htbp]
\centering
\includegraphics[width=12cm]{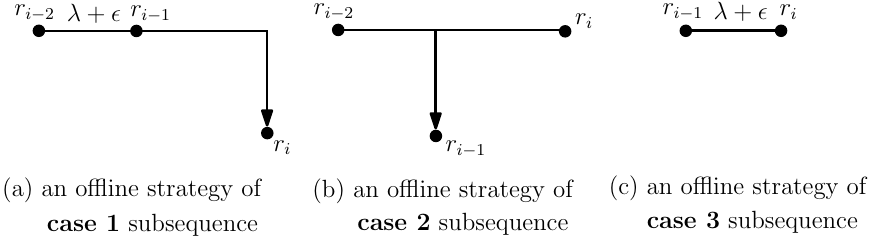}
\caption{\label{subseq} Offline strategies of \textbf{case 1/2/3} subsequences}
\end{figure}

\textbf{Case 2:} A subsequence $\langle r_{i-1}, r_i \rangle$ of two successive \textbf{Type-$K_{1}$} requests. By definition, $r_{i-2}$ and $r_i$ are at the same server, while $r_{i-1}$ is at a different server. In an offline strategy, $s[r_i]$ can keep its copy throughout the period $(t_{i-2}, t_i)$ and serve $r_{i-1}$ by a transfer (see Figure \ref{subseq}(b)). Hence, $\textbf{offline}(r_{i-2}, r_i)\leq (t_i-t_{i-2})+\trfrcost$.

(1) If $r_i$ is a \textbf{Type-$K_{1a}$} request, we have $\textbf{online}(r_{i-1}, r_{i})\geq(t_i-t_{i-1})+\trfrcost$ by Observation \ref{interobs} and $t_i-t_{i-2}=t_i-t_{p(i)}=\trfrcost+\epsilon$ by Observation \ref{requestk2}. Since $r_{i-1}$ is a \textbf{Type-$K_{1}$} request, by Observation \ref{interobs}, we always have $\textbf{online}(r_{i-2}, r_{i-1})\geq(t_{i-1}-t_{i-2})+\trfrcost$. Hence, $\textbf{online}(r_{i-2}, r_{i})=\textbf{online}(r_{i-2}, r_{i-1})+\textbf{online}(r_{i-1}, r_{i})\geq(t_i-t_{i-2})+2\trfrcost$. Thus, $\frac{\textbf{online}(r_{i-2}, r_{i})}{\textbf{offline}(r_{i-2}, r_i)}\geq\frac{(t_i-t_{i-2})+2\trfrcost}{(t_i-t_{i-2})+\trfrcost}=\frac{3\trfrcost+\epsilon}{2\trfrcost+\epsilon}$, which approaches $\frac{3}{2}$ as $\epsilon \rightarrow 0$.

(2) If $r_i$ is a \textbf{Type-$K_{1b}$} or \textbf{Type-$K_{1c}$} request, we have $\textbf{online}(r_{i-1}, r_{i})\geq 2(t_i-t_{i-1})+\trfrcost-\epsilon$ by Observation \ref{interobs}. Since $r_{i-1}$ is a \textbf{Type-$K_{1}$} request, by Observation \ref{interobs}, we always have $\textbf{online}(r_{i-2}, r_{i-1})\geq(t_{i-1}-t_{i-2})+\trfrcost$, and by Observation \ref{requestk2}, we have $t_{i-1}-t_{i-2}<\trfrcost+\epsilon$. Hence, $\textbf{online}(r_{i-2}, r_{i})=\textbf{online}(r_{i-2}, r_{i-1})+\textbf{online}(r_{i-1}, r_{i})\geq(t_i-t_{i-2})+(t_{i}-t_{i-1})+2\trfrcost-\epsilon$. Thus, $\frac{\textbf{online}(r_{i-2}, r_{i})}{\textbf{offline}(r_{i-2}, r_i)}\geq\frac{(t_i-t_{i-2})+(t_{i}-t_{i-1})+2\trfrcost-\epsilon}{(t_i-t_{i-2})+\trfrcost}=1+\frac{(t_{i}-t_{i-1})+\trfrcost-\epsilon}{(t_i-t_{i-1})+(t_{i-1}-t_{i-2})+\trfrcost}>1+\frac{(t_{i}-t_{i-1})+\trfrcost-\epsilon}{(t_i-t_{i-1})+2\trfrcost+\epsilon}$, which approaches at least $\frac{3}{2}$ as $\epsilon \rightarrow 0$.

\textbf{Case 3:} A subsequence $\langle r_i \rangle$ of a single \textbf{Type-$K_{2}$} request. By definition, $r_{i-1}$ and $r_{i}$ are at the same server. By Observation \ref{requestk1}, $t_{i}-t_{i-1}=\trfrcost+\epsilon$. In an offline strategy, $s[r_i]$ can keep its copy throughout the period $(t_{i-1}, t_i)$ and serve $r_{i}$ locally (see Figure \ref{subseq}(c)). Hence, $\textbf{offline}(r_{i-1}, r_i)\leq t_i-t_{i-1}=\trfrcost+\epsilon$. By Observation \ref{interobs}, $\textbf{online}(r_{i-1}, r_{i})\geq 2\trfrcost+\epsilon$. Thus, $\frac{\textbf{online}(r_{i-1}, r_{i})}{\textbf{offline}(r_{i-1}, r_i)}\geq\frac{2\trfrcost+\epsilon}{\trfrcost+\epsilon}$, which approaches $2$ as $\epsilon \rightarrow 0$.

In summary, the cost ratio between the online algorithm and the offline strategy for any \textbf{case 1/2/3} subsequence is at least $\frac{3}{2}$ or can be made arbitrarily close to $\frac{3}{2}$ as $\epsilon \rightarrow 0$.

Any request sequence $\langle r_0, r_1, r_2, \dots, r_m \rangle$ generated by the adversary can be partitioned into \textbf{case 1/2/3} subsequences starting from $r_1$ or $r_2$. Specifically, we can identify maximal segments of successive \textbf{Type-$K_1$} requests (called maximal \textbf{Type-$K_1$} segments for short). A maximal \textbf{Type-$K_1$} segment means that the requests immediately preceding and succeeding the segment are both \textbf{Type-$K_2$} requests. If a maximal \textbf{Type-$K_1$} segment has an even number of \textbf{Type-$K_1$} requests, we can further break it into \textbf{case 2} subsequences. If a maximal \textbf{Type-$K_1$} segment has an odd number of \textbf{Type-$K_1$} requests, we can pair the first \textbf{Type-$K_1$} request with the preceding \textbf{Type-$K_2$} request to form a \textbf{case 1} subsequence and break the remaining segment into \textbf{case 2} subsequences. For the rest \textbf{Type-$K_2$} requests, each of them constitutes a \textbf{case 3} subsequence on its own. With the above partitioning, we are left with only $\langle r_0, r_1 \rangle$ or $\langle r_0, r_1, r_2 \rangle$ at the beginning of the request sequence, where $\langle r_0, r_1, r_2 \rangle$ happens only when $r_2$ is at the head of a maximal \textbf{Type-$K_1$} segment of odd length. 

The aforesaid offline strategies constructed for individual \textbf{case 1/2/3} subsequences can be concatenated to form an offline strategy for all the requests after $r_1$ or $r_2$. Thus, by the above analysis, $\frac{\textbf{online}(r_1, r_m)}{\textbf{offline}(r_1, r_m)}$ or $\frac{\textbf{online}(r_2, r_m)}{\textbf{offline}(r_2, r_m)}$ approaches at least $\frac{3}{2}$ as $\epsilon \rightarrow 0$, where $\textbf{offline}(r_{i}, r_{j})$ denotes the total storage and transfer cost incurred over the period $(t_{i}, t_j]$ in the offline strategy constructed. Next, let us consider the requests left at the beginning. In the case of $\langle r_0, r_1 \rangle$, an offline strategy can perform a transfer at time $t_1$ to serve $r_1$ and then both servers hold data copies right after $r_1$ for whatever replication strategy to carry on (see Figure \ref{r0r1}). As a result, $\textbf{offline}(r_0, r_1) = \trfrcost+\epsilon$. Consequently, $\frac{\textbf{online}(r_0, r_m)}{\textbf{offline}(r_0, r_m)} \geq \frac{\textbf{online}(r_1, r_m)}{\textbf{offline}(r_0, r_m)} = \frac{\textbf{online}(r_1, r_m)}{\textbf{offline}(r_1, r_m)+\trfrcost+\epsilon}$, which approaches at least $\frac{3}{2}$ as $\epsilon \rightarrow 0$ and $m \rightarrow \infty$. 
Similarly, in the case of $\langle r_0, r_1, r_2 \rangle$, an offline strategy can hold data copies at both servers after $r_1$ until $r_2$ and then they both have copies right after $r_2$ for whatever replication strategy to carry on (see Figure \ref{r0r1r2}). As a result, $\textbf{offline}(r_0, r_2) < 3\trfrcost+3\epsilon$.\footnote{Since $r_2$ is a \textbf{Type-$K_1$} request, by Observation \ref{requestk2}, the time duration between $r_1$ and $r_2$ is shorter than $\trfrcost+\epsilon$.} Hence, $\frac{\textbf{online}(r_0, r_m)}{\textbf{offline}(r_0, r_m)} \geq \frac{\textbf{online}(r_2, r_m)}{\textbf{offline}(r_0, r_m)} > \frac{\textbf{online}(r_2, r_m)}{\textbf{offline}(r_2, r_m)+3\trfrcost+3\epsilon}$, which also approaches at least $\frac{3}{2}$ as $\epsilon \rightarrow 0$ and $m \rightarrow \infty$. Therefore, the consistency of any deterministic learning-augmented algorithm must be at least $\frac{3}{2}$. 

\begin{figure}[htbp]
\centering
\includegraphics[width=0.9cm]{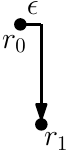}
\caption{\label{r0r1}
An offline strategy for $\langle r_0, r_1 \rangle$}
\end{figure}
\begin{figure}[htbp]
\centering
\includegraphics[width=3cm]{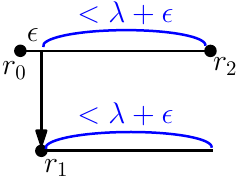}
\caption{\label{r0r1r2}
An offline strategy for $\langle r_0, r_1, r_2 \rangle$}
\end{figure}

\section{Experimental Evaluation}
\label{experiment}
\subsection{Experimental Setup}
We conduct simulation experiments to evaluate our proposed
algorithms, using the object storage traces provided by IBM (available at https://www.ibm.com/cloud/blog/object-storage-traces). The traces include read and write requests made to objects in a cloud-based
object storage over 7 days in the year of 2019. We select several objects and filter out their write operations. For each object, we randomly distribute its requests over 10 different servers following a Zipf distribution, where each request is assigned to a server indexed by $i$ with probability $i^{-1}\slash\sum_{j=1}^{10}j^{-1}$ ($i=1,2,\dots,10$). The storage cost rate of each server is assumed to be 1 per second. The predictions of inter-request times are randomly generated according to the ground truth and a specified prediction accuracy. We obtain the online costs of serving the request sequence by running our proposed algorithm with different combinations of hyper-parameter $\hyperpar\in\left\{0, 0.1,0.2,\dots,1\right\}$, transfer cost $\trfrcost\in\left\{10, 100, 1000, 10000\right\}$ and prediction accuracy of $\left\{0\%, 10\%,20\%,\dots,100\%\right\}$. We normalize these online costs by the optimal offline costs derived from dynamic programming \cite{wang2018cost},
and produce 3-dimensional plots to illustrate the algorithm performance under different settings. Similar performance trends are observed for the objects that we experiment with. Thus, we present the representative results of one particular object (with identifier ``652aaef228286e0a'') which has 11688 read requests over 7 days.

\subsection{Experimental Results}
First, we study the original Algorithm \ref{alg3}. Figures \ref{2} to \ref{5} show the online-to-optimal cost ratios for different settings (note that the vertical axis has different scales in different figures). All the ratios are bounded by $1+\frac{1}{\hyperpar}$ (robustness derived in Section \ref{sec:robustness}), and the ratios of those 100\% prediction accuracy settings are all bounded by $\frac{5+\hyperpar}{3}$ (consistency derived in Section \ref{sec:consistency}). These results verify our robustness and consistency analysis.

\begin{figure*}[htbp]
\centering
\begin{minipage}[t]{0.24\textwidth}
\centering
\includegraphics[width=1\textwidth]{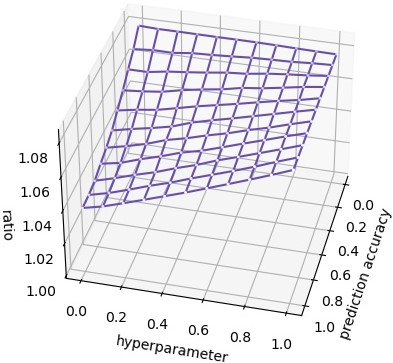}
\caption{\label{2} $\trfrcost = 10$}
\end{minipage}
\begin{minipage}[t]{0.24\textwidth}
\centering
\includegraphics[width=1\textwidth]{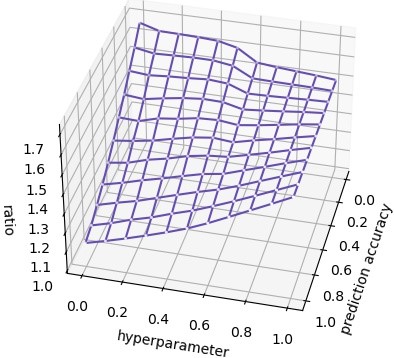}
\caption{\label{3} $\trfrcost = 100$}
\end{minipage}
\begin{minipage}[t]{0.24\textwidth}
\centering
\includegraphics[width=1\textwidth]{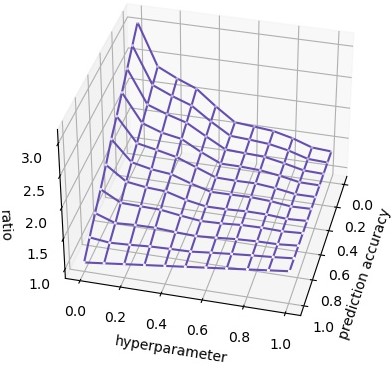}
\caption{\label{4} $\trfrcost = 1000$}
\end{minipage}
\begin{minipage}[t]{0.24\textwidth}
\centering
\includegraphics[width=1\textwidth]{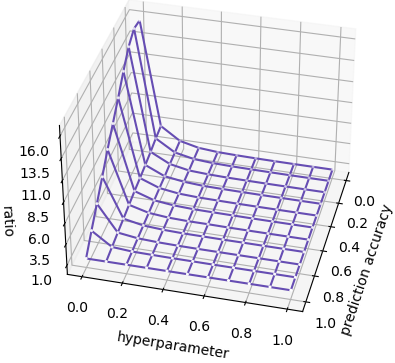}
\caption{\label{5} $\trfrcost = 10000$}
\end{minipage}
\end{figure*}

\begin{figure*}[htbp]
\centering
\begin{minipage}[t]{0.24\textwidth}
\centering
\includegraphics[width=1\textwidth]{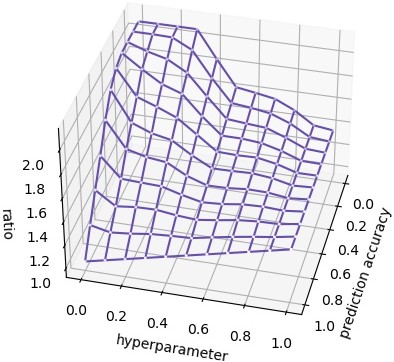}
\caption{\label{adapt1}
$\trfrcost = 1000, \beta = 0.1$}
\end{minipage}
\centering
\begin{minipage}[t]{0.24\textwidth}
\centering
\includegraphics[width=1\textwidth]{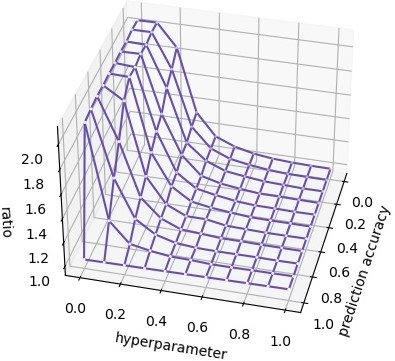}
\caption{\label{adapt2}
$\trfrcost = 10000, \beta = 0.1$}
\end{minipage}
\begin{minipage}[t]{0.24\textwidth}
\centering
\includegraphics[width=1\textwidth]{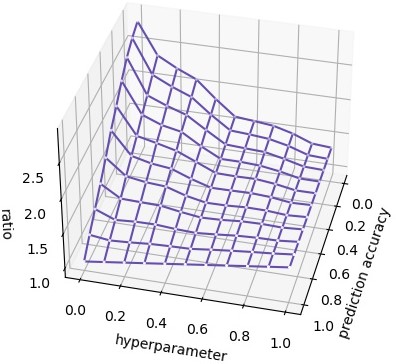}
\caption{\label{adapt3}
$\trfrcost = 1000, \beta = 1$}
\end{minipage}
\begin{minipage}[t]{0.24\textwidth}
\centering
\includegraphics[width=1\textwidth]{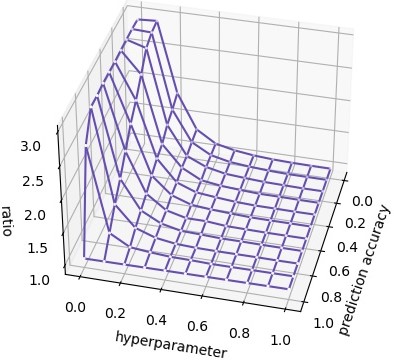}
\caption{\label{adapt4}
$\trfrcost = 10000, \beta = 1$}
\end{minipage}
\end{figure*}

One common feature of all the figures is that when the hyper-parameter $\hyperpar=1$, the online-to-optimal cost ratio is a constant, no matter what the prediction accuracy is. The reason is that when $\hyperpar=1$, the proposed algorithm does not make use of predictions to adjust the durations of regular copies and thus its performance is independent of the prediction accuracy. Another common feature is that the online-to-optimal cost ratio attains its minimum when the hyper-parameter $\alpha$ gets close to 0 and the prediction accuracy goes towards 100\%. This is because when the predictions forecast future request arrivals correctly and the proposed algorithm makes almost full use of the predictions, the online cost would be reduced. This demonstrates the effectiveness of our proposed algorithm when incorporating accurate predictions. 

Since there are 11688 requests over 7 days distributed among 10 servers, the average inter-request
time at a server is about 500 seconds. When the transfer cost $\trfrcost$ is very low (Figure \ref{2}), inter-request times are generally much greater than $\trfrcost$. Hence, in the optimal offline strategy, usually only one data copy is maintained in the system, and most requests are served by transfers. By our online algorithm, the regular copies have intended durations bounded by $\trfrcost$ irrespective of the prediction accuracy, and thus most requests are also served by transfers. Since inter-request times are much greater than $\trfrcost$, one special copy is typically maintained in the system. Therefore, our online algorithm behaves similarly to the optimal offline strategy. As a result, the online-to-optimal cost ratio is close to 1. 

When the transfer cost $\trfrcost$ is higher (Figures \ref{3}, \ref{4} and \ref{5}), the online-to-optimal cost ratio reaches its peak when both the hyper-parameter $\hyperpar$ and prediction accuracy approach 0. The reason is that as $\trfrcost$ grows, more inter-request times become smaller than $\trfrcost$, and hence more requests are served by local copies in the optimal offline strategy. For those requests $r_i$'s where $\hyperpar\trfrcost < t_i-t_{p(i)}\leq \trfrcost$, mispredictions would lead to them served by transfers by our online algorithm, which incurs unnecessary transfer cost. Therefore, the performance of our online algorithm deteriorates as the prediction accuracy and $\hyperpar$ decrease. When the transfer cost $\trfrcost$ is very high (Figure \ref{5}), there is almost no difference among the online-to-optimal cost ratios for different settings, unless both $\hyperpar$ and prediction accuracy approach 0. In most settings, the ratios are close to 1. This is because most inter-request times are less than $\hyperpar\trfrcost$ in this case. By our online algorithm, the regular copies have intended durations at least $\hyperpar\trfrcost$ irrespective of the prediction accuracy, and thus most requests are served by local copies, consistent with the optimal offline strategy.

Next, we study the adapted Algorithm \ref{alg3} to achieve a robustness of $2+\beta$.
We run the original Algorithm \ref{alg3} for the first $100$ requests of the trace as a warm-up to get the initial $\textbf{OnlineU}$ and $\textbf{OPTL}$. 
After that, we run the adapted Algorithm \ref{alg3} to decide whether to set the intended duration of the regular copy after each request always to $\trfrcost$ or according to the prediction. 
Figures \ref{adapt1} to \ref{adapt4} show the online-to-optimal cost ratios for $\trfrcost = 1000, 10000$ and $\beta = 0.1, 1$. As can be seen, the adapted Algorithm \ref{alg3} successfully prevents the cost ratio from growing beyond the target robustness. This demonstrates the effectiveness of our adaptation to deal with terrible predictions. The results for $\trfrcost = 10, 100$ are the same as those for the original Algorithm \ref{alg3} (hence not repeated here) because the online-to-optimal cost ratios are well below the target robustness $2+\beta$.

\end{document}